\documentclass[11pt]{article}%
\usepackage{amsmath}
\usepackage{amssymb}
\usepackage{amsfonts}

\usepackage{cite}
\usepackage{graphicx}
\usepackage{float}
\usepackage{longtable}
\usepackage[utf8]{inputenc}
\usepackage{hyperref}

\usepackage{algorithmicx}
\usepackage[ruled,vlined]{algorithm2e}
\SetAlgoCaptionSeparator{ }

\usepackage{fullpage}

\newcommand{\fref}[1]{Fig.~\ref{#1}}

\newcommand{\tref}[1]{Table~\ref{#1}}
\newcommand{\sref}[1]{Section~\ref{#1}}

\providecommand{\U}[1]{\protect\rule{.1in}{.1in}}
\newtheorem{theorem}{Theorem}

\newtheorem{problem}{Problem}

\newenvironment{proof}[1][Proof]{\textbf{#1.} }{\ \rule{0.5em}{0.5em}}

\begin{document}

\title{Optimizing High-Efficiency Quantum Memory with Quantum Machine Learning for Near-Term Quantum Devices}
\author{Laszlo Gyongyosi\thanks{School of Electronics and Computer Science, University of Southampton, Southampton SO17 1BJ, U.K., and Department of Networked Systems and Services, Budapest University of Technology and Economics, 1117 Budapest, Hungary, and MTA-BME Information Systems Research Group, Hungarian Academy of Sciences, 1051 Budapest, Hungary.}
\and Sandor Imre\thanks{Department of Networked Systems and Services, Budapest University of Technology and Economics, 1117 Budapest, Hungary.}
}

\date{}

\maketitle
\begin{abstract}
Quantum memories are a fundamental of any global-scale quantum Internet, high-performance quantum networking and near-term quantum computers. A main problem of quantum memories is the low retrieval efficiency of the quantum systems from the quantum registers of the quantum memory. Here, we define a novel quantum memory called high-retrieval-efficiency (HRE) quantum memory for near-term quantum devices. An HRE quantum memory unit integrates local unitary operations on its hardware level for the optimization of the readout procedure and utilizes the advanced techniques of quantum machine learning. We define the integrated unitary operations of an HRE quantum memory, prove the learning procedure, and evaluate the achievable output signal-to-noise ratio values. We prove that the local unitaries of an HRE quantum memory achieve the optimization of the readout procedure in an unsupervised manner without the use of any labeled data or training sequences. We show that the readout procedure of an HRE quantum memory is realized in a completely blind manner without any information about the input quantum system or about the unknown quantum operation of the quantum register. We evaluate the retrieval efficiency of an HRE quantum memory and the output SNR (signal-to-noise ratio). The results are particularly convenient for gate-model quantum computers and the near-term quantum devices of the quantum Internet.
\end{abstract}

\section{Introduction}
\label{sec1}
Quantum memories are a fundamental of any global-scale quantum Internet \cite{puj1,puj3,ref1,ref2,ref3, refn7}. However, while quantum repeaters can be realized without the necessity of quantum memories \cite{puj1,puj3}, these units, in fact, are required for guaranteeing an optimal performance in any high-performance quantum networking scenario \cite{ref5,ref6,ref7,ref11,ref13a,ref13,ref13b,ref8,ref9,ref10,add1,add2,add3,refqirg,ref18,ref19,ref20,ref21,add4,refn7,refn5,refn3,sat,telep,refn1,refn2,refn4,refn6,puj1}. Therefore, the utilization of quantum memories still represents a fundamental problem in the quantum Internet \cite{ref23,ref24,ref25,ref26,ref27,nadd1,nadd2,nadd3,nadd4,nadd5}, since the near-term quantum devices (such as quantum repeaters \cite{ref1,ref3,ref6,refn6,ref50,ref51,ref53,ref56,ref58}) and gate-model quantum computers \cite{qc1,qc2,qc3,qc4,qc5,qc6,qcadd1,qcadd2,qnuj1,qnuj2,qnuj3,qnuj4} have to store the quantum states in their local quantum memories \cite{ref29,ref34,ref35,ref36,ref37,ref38,ref39,ref40,ref41,ref42,ref43,ref44,ref45,ref46,ref47,ref48,ref49,ref50,ref51,ref52,ref53,ref54,ref55,ref56,ref57,ref58,ref59,ref60,ref61,ref62
}. The main problem here is the efficient readout of the stored quantum systems and the low retrieval efficiency of these systems from the quantum registers of the quantum memory. Currently, no general solution to this problem is available, since the quantum register evolves the stored quantum systems via an unknown operation, and the input quantum system is also unknown, in a general scenario \cite{refn7,ref1,ref5,ref6,ref7,ref13,ref13a}. The optimization of the readout procedure is therefore a hard and complex problem. Several physical implementations have been developed in the last few years \cite{qm1,qm2,qm3,qm4,qm5,qmb1,qmb2,qmb3,qmb4,qmb5,qmb6,qmb7,qmb8,qmb9,qmb10,qmb11,qmb12,qmb13,qmb14,qmb15,qmb16}. However, these experimental realizations have several drawbacks, in general because the output signal-to-noise ratio (SNR) values are still not satisfactory for the construction of a powerful, global-scale quantum communication network. As another important application field in quantum communication, the methods of quantum secure direct communication \cite{radd1,radd2,radd3,radd4} also require quantum memory.

Here, we define a novel quantum memory called high-retrieval-efficiency (HRE) quantum memory for near-term quantum devices. An HRE quantum memory unit integrates local unitary operations on its hardware level for the optimization of the readout procedure. An HRE quantum memory unit utilizes the advanced techniques of quantum machine learning \cite{ref31,ref32,ref33} to achieve a significant improvement in the retrieval efficiency. We define the integrated unitary operations of an HRE quantum memory, prove the learning procedure, and evaluate the achievable output SNR values. The local unitaries of an HRE quantum memory achieve the optimization of the readout procedure in an unsupervised manner without the use of any labeled data or any training sequences. The readout procedure of an HRE quantum memory is realized in a completely blind manner. It requires no information about the input quantum system or about the quantum operation of the quantum register. (It is motivated by the fact that this information is not accessible in any practical setting.)

The proposed model assumes that the main challenge is the recovery the stored quantum systems from the quantum register of the quantum memory unit, such that both the input quantum system and the transformation of the quantum memory are unknown. The optimization problem of the readout process also integrates the efficiency of the write-in procedure. In the proposed model, the noise and uncertainty added by the write-in procedure are included in the unknown transformation of the $QR$ quantum register of the quantum memory that results in a $\sigma _{QR}$ mixed quantum system in $QR$.

The novel contributions of our manuscript are as follows:
\begin{enumerate}
\item We define a novel quantum memory called high-retrieval-efficiency (HRE) quantum memory. 
\item An HRE quantum memory unit integrates local unitary operations on its hardware level for the optimization of the readout procedure and utilizes the advanced techniques of quantum machine learning.
\item We define the integrated unitary operations of an HRE quantum memory, prove the learning procedure, and evaluate the achievable output signal-to-noise ratio values. We prove that local unitaries of an HRE quantum memory achieve the optimization of the readout procedure in an unsupervised manner without the use of any labeled data or training sequences.
\item We evaluate the retrieval efficiency of an HRE quantum memory and the output SNR.
\item The proposed results are convenient for gate-model quantum computers and near-term quantum devices.
\end{enumerate}

This paper is organized as follows. \sref{sec2} defines the system model and the problem statement. \sref{sec3} evaluates the integrated local unitary operations of an HRE quantum memory. \sref{sec4} proposes the retrieval efficiency in terms of the achievable output SNR values. Finally, \sref{sec5} concludes the results. Supplemental material is included in the Appendix.

\section{System Model and Problem Statement}
\label{sec2}
\subsection{System Model }
Let $\rho _{in} $ be an unknown input quantum system formulated by $n$ unknown density matrices, 
\begin{equation} \label{ZEqnNum367280} 
\rho _{in} =\sum _{i=1}^{n}\lambda _{i}^{\left(in\right)} {\left| \psi _{i}  \right\rangle} {\left\langle \psi _{i}  \right|}  ,  
\end{equation} 
where $\lambda _{i}^{\left(in\right)} \ge 0$, and $\sum _{i=1}^{n}\lambda _{i}^{\left(in\right)}  =1$. 

The input system is received and stored in the $QR$ quantum register of the HRE quantum memory unit. The quantum systems are $d$-dimensional systems ($d=2$ for a qubit system). For simplicity, we focus on $d=2$ dimensional quantum systems throughout the derivations.

The $U_{QR} $ unknown evolution operator of the $QR$ quantum register defines a mixed state $\sigma _{QR} $ as 
\begin{equation} \label{ZEqnNum669290} 
\begin{split}
   {{\sigma }_{QR}}&={{U}_{QR}}{{\rho }_{in}}U_{QR}^{\dagger } \\ 
 & =\sum\limits_{i=1}^{n}{{{\lambda }_{i}}\left| {{\varphi }_{i}} \right\rangle \left\langle  {{\varphi }_{i}} \right|,}  
\end{split}
\end{equation} 
where $\lambda _{i} \ge 0$, $\sum _{i=1}^{n}\lambda _{i}  =1$. 

Let us allow to rewrite \eqref{ZEqnNum669290} for a particular time $t$, $t=1,\ldots ,T$, where $T$ is a total evolution time, via a mixed system $\sigma _{QR}^{\left(t\right)} $, as
\begin{equation} \label{ZEqnNum694160} 
\begin{split}
   \sigma _{QR}^{\left( t \right)}&=U_{QG}^{\left( t \right)}{{\rho }_{in}}{{\left( U_{QG}^{\left( t \right)} \right)}^{\dagger }} \\ 
 & =\sum\limits_{i=1}^{n}{\lambda _{i}^{\left( t \right)}\left| \varphi _{i}^{\left( t \right)} \right\rangle \left\langle  \varphi _{i}^{\left( t \right)} \right|} \\ 
 & =\sum\limits_{i=1}^{n}{\left( \sqrt{\lambda _{i}^{\left( t \right)}}\left| \varphi _{i}^{\left( t \right)} \right\rangle  \right)\left( \sqrt{\lambda _{i}^{\left( t \right)}}\left\langle  \varphi _{i}^{\left( t \right)} \right| \right)} \\ 
 & =\sum\limits_{i=1}^{n}{X_{i}^{\left( t \right)}{{\left( X_{i}^{\left( t \right)} \right)}^{\dagger }}} \\ 
 & ={{X}^{\left( t \right)}}{{\left( {{X}^{\left( t \right)}} \right)}^{\dagger }},  
\end{split}
\end{equation} 
where $U_{QR}^{\left(t\right)} $ is an unknown evolution matrix of the $QR$ quantum register at a given $t$, with a dimension
\begin{equation} \label{4)} 
\dim \left(U_{QR}^{\left(t\right)} \right)=d^{n} \times d^{n} ,   
\end{equation} 
with $0\le \lambda _{i}^{\left(t\right)} \le 1$, $\sum _{i}\lambda _{i}^{\left(t\right)}  =1$, while $X_{i}^{\left(t\right)} \in {\mathbb{C}}$ is an unknown complex quantity, defined as
\begin{equation} \label{ZEqnNum880602} 
X_{i}^{\left(t\right)} =\sqrt{\lambda _{i}^{\left(t\right)} } {\left| \varphi _{i}^{\left(t\right)}  \right\rangle}  
\end{equation} 
and
\begin{equation} \label{6)} 
X^{\left(t\right)} =\sum _{i=1}^{n}X_{i}^{\left(t\right)}  .  
\end{equation} 
Then, let us rewrite $\sigma _{QR}^{\left(t\right)} $ from \eqref{ZEqnNum694160} as
\begin{equation} \label{ZEqnNum275626} 
\sigma _{QR}^{\left(t\right)} =\rho _{in} +\zeta _{QR}^{\left(t\right)} , 
\end{equation} 
where $\rho _{in} $ is as in \eqref{ZEqnNum367280}, and $\zeta _{QR}^{\left(t\right)} $ is an unknown residual density matrix at a given $t$. 

Therefore, \eqref{ZEqnNum275626} can be expressed as a sum of $M$ source quantum systems,
\begin{equation} \label{ZEqnNum522860} 
\sigma _{QR}^{\left(t\right)} =\sum _{m=1}^{M}\rho _{m}  ,  
\end{equation} 
where $\rho _{m} $ is the $m$-th source quantum system and $m=1,\ldots ,M$, where  
\begin{equation} \label{ZEqnNum656436} 
M=2,   
\end{equation} 
in our setting, since  
\begin{equation} \label{10)} 
\rho _{1} =\rho _{in}  
\end{equation} 
and
\begin{equation} \label{11)} 
\rho _{2} =\zeta _{QR}^{\left(t\right)} . 
\end{equation} 
In terms of the $M$ subsystems, \eqref{ZEqnNum694160} can be rewritten as
\begin{equation} \label{ZEqnNum630316} 
\begin{split}
  \sigma _{QR}^{\left( t \right)}&=\sum\limits_{m=1}^{M}{\sum\limits_{i=1}^{n}{\lambda _{i}^{\left( m,t \right)}\left| \varphi _{i}^{\left( m,t \right)} \right\rangle \left\langle  \varphi _{i}^{\left( m,t \right)} \right|}} \\ 
 & =\sum\limits_{m=1}^{M}{\sum\limits_{i=1}^{n}{\sqrt{\lambda _{i}^{\left( m,t \right)}}\left| \varphi _{i}^{\left( m,t \right)} \right\rangle \sqrt{\lambda _{i}^{\left( m,t \right)}}\left\langle  \varphi _{i}^{\left( m,t \right)} \right|}} \\ 
 & =\sum\limits_{m=1}^{M}{\sum\limits_{i=1}^{n}{X_{i}^{\left( m,t \right)}{{\left( X_{i}^{\left( m,t \right)} \right)}^{\dagger }}}} \\ 
 & =\sum\limits_{m=1}^{M}{{{X}^{\left( m,t \right)}}{{\left( {{X}^{\left( m,t \right)}} \right)}^{\dagger }}},  
\end{split}
\end{equation} 
where $X_{i}^{\left(m,t\right)} $ is a complex quantity associated with an $m$-th source system, 
\begin{equation} \label{ZEqnNum718067} 
X_{i}^{\left(m,t\right)} =\sqrt{\lambda _{i}^{\left(m,t\right)} } {\left| \varphi _{i}^{\left(m,t\right)}  \right\rangle} , 
\end{equation} 
with $0\le \lambda _{i}^{\left(m,t\right)} \le 1$, $\sum _{m}\sum _{i}\lambda _{i}^{\left(m,t\right)}   =1$, and
\begin{equation} \label{ZEqnNum440486} 
X^{\left(m,t\right)} =\sum _{i=1}^{n}X_{i}^{\left(m,t\right)}  .  
\end{equation} 
The aim is to find the $V_{QG} $ inverse matrix of the unknown evolution matrix $U_{QR} $ in \eqref{ZEqnNum669290}, as
\begin{equation} \label{15)} 
V_{QG} =U_{QG}^{-1} ,  
\end{equation} 
that yields the separated readout quantum system of the HRE quantum memory unit for $t=1,\ldots ,T$, such that for a given $t$,
\begin{equation} \label{16)} 
\sigma _{out}^{\left(t\right)} =V_{QG}^{\left(t\right)} \sigma _{QR}^{\left(t\right)} \left(V_{QG}^{\left(t\right)} \right)^{\dag } ,  
\end{equation} 
where 
\begin{equation} \label{17)} 
V_{QG}^{\left(t\right)} =\left(U_{QG}^{\left(t\right)} \right)^{-1} .  
\end{equation} 
For a total evolution time $T$, the target $\sigma _{out} $ density matrix is yielded at the output of the HRE quantum memory unit, as
\begin{equation} \label{18)} 
\sigma _{out} \approx \sum _{i=1}^{n}\lambda _{i}^{\left(in\right)} {\left| \psi _{i}  \right\rangle} {\left\langle \psi _{i}  \right|}   
\end{equation} 
with a sufficiently high SNR value, 
\begin{equation} \label{19)} 
{\rm SNR}\left(\sigma _{out} \right)\ge x,   
\end{equation} 
where $x$ is an SNR value that depends on the actual physical layer attributes of the experimental implementation.

The problem is therefore that both the input quantum system \eqref{ZEqnNum367280} and the transformation matrix $U_{QR} $ in \eqref{ZEqnNum669290} of the quantum register are unknown. As we prove, by integrating local unitaries to the HRE quantum memory unit, the unknown evolution matrix of the quantum register can be inverted, which allows us to retrieve the quantum systems of the quantum register. The retrieval efficiency will be also defined in a rigorous manner.

\subsection{Problem Statement}
The problem statement is as follows.

Let $M$ be the number of source systems in the $QR$ quantum register such that the sum of the $M$ source systems identifies the mixed state of the quantum register. Let $m$ be the index of the source system, $m=1,\ldots ,M$, such that $m=1$ identifies the unknown input quantum system stored in the quantum register (target source system), while $m=2,\ldots ,M$ are some unknown residual quantum systems. The input quantum system, the residual systems, and the transformation operation of the quantum register are unknown. The aim is then to define local unitary operations to be integrated on the HRE quantum memory unit for an HRE readout procedure in an unsupervised manner with unlabeled data.

The problems to be solved are summarized in Problems 1--4.
\begin{problem}
Find an unsupervised quantum machine learning method, $U_{ML} $, for the factorization of the unknown mixed quantum system of the quantum register via a blind separation of the unlabeled quantum register. Decompose the unknown mixed system state into a basis unitary and a residual quantum system.
\end{problem}
\begin{problem}
Define a unitary operation for partitioning the bases with respect to the source systems of the quantum register.
\end{problem}
\begin{problem}
Define a unitary operation for the recovery of the target source system.
\end{problem}
\begin{problem}
Evaluate the retrieval efficiency of the HRE quantum memory in terms of the achievable SNR.
\end{problem}

The resolutions of the problems are proposed in Theorems 1--4.

The schematic model of an HRE quantum memory unit is depicted in \fref{fig1}. 

\begin{center}
\begin{figure*}[!htbp]
\begin{center}
\includegraphics[angle = 0,width=1\linewidth]{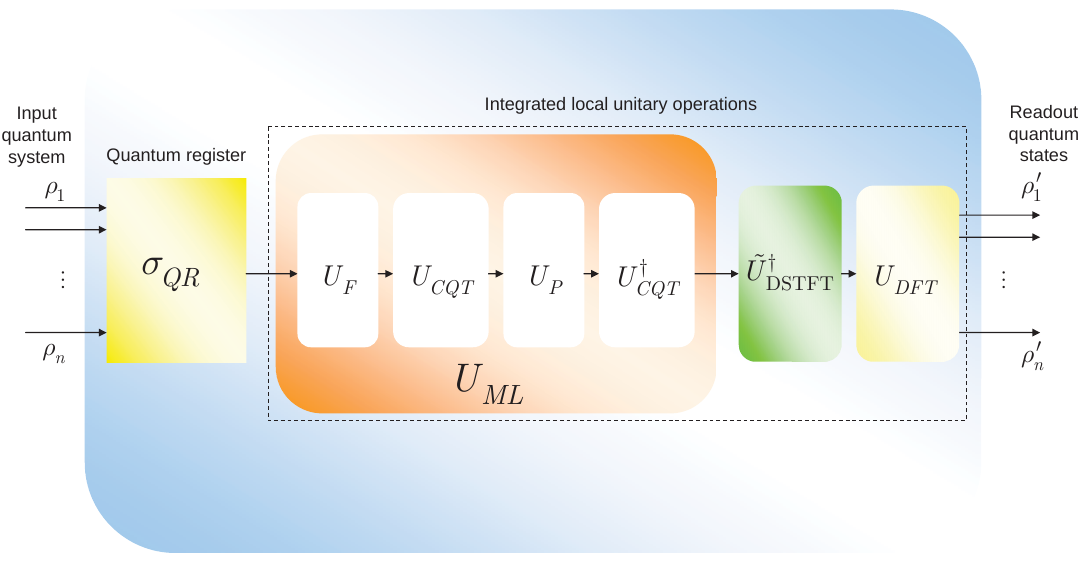}
\caption{The schematic model of a high-retrieval-efficiency (HRE) quantum memory unit. The HRE quantum memory unit contains a $QR$ quantum register and integrated local unitary operations. The $n$ input quantum systems, $\rho _{1} \ldots \rho _{n} $, are received and stored in the quantum register. The state of the $QR$ quantum register defines a mixed state, $\sigma _{QR} =\sum _{i}\lambda _{i} \rho _{i}  $, where $\sum _{i}\lambda _{i}  =1$. The stored density matrices of the $QR$ quantum register are first transformed by a $U_{ML} $, a quantum machine learning unitary (depicted by the orange-shaded box) that implements an unsupervised learning for a blind separation of the unlabeled input, and decomposable as $U_{ML} =U_{F} U_{CQT} U_{P} U_{CQT}^{\dag } $, where $U_{F} $ is a factorization unitary, $U_{CQT} $ is the quantum constant $Q$ transform with a windowing function $f_{W} $ for the localization of the wave functions of the quantum register, $U_{P} $ is a basis partitioning unitary, while $U_{CQT}^{\dag } $ is the inverse of $U_{CQT} $. The result of $U_{ML} $ is processed further by the $\tilde{U}_{{\rm DSTFT}}^{\dag } $ unitary (depicted by the green-shaded box) that realizes the inverse quantum discrete short-time Fourier transform (DSTFT) operation (depicted by the yellow-shaded box), and by the $U_{DFT} $ (quantum discrete Fourier transform) unitary to yield the desired output $\rho '_{1} \ldots \rho '_{n} $.} 
 \label{fig1}
 \end{center}
\end{figure*}
\end{center}

The procedures realized by the integrated unitary operations of the HRE quantum memory are depicted in \fref{fig4}.

 \begin{center}
\begin{figure*}[!htbp]
\begin{center}
\includegraphics[angle = 0,width=1\linewidth]{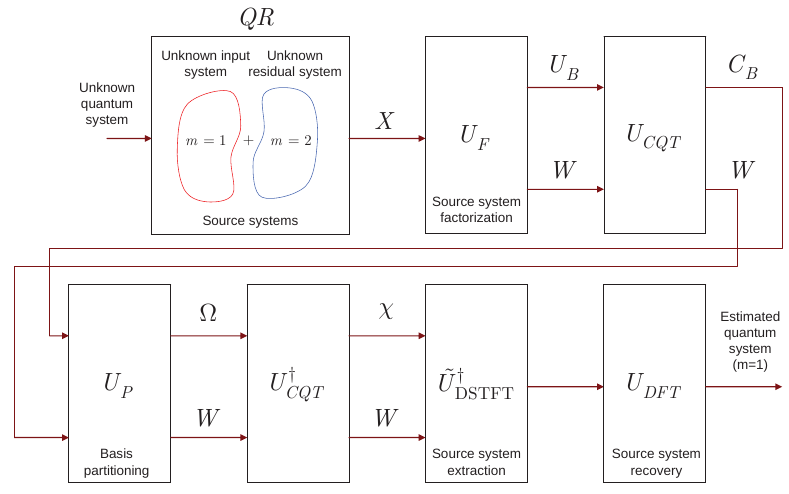}
\caption{Detailed procedures of an HRE quantum memory. The unknown input quantum system is stored in the $QR$ quantum register that realizes an unknown transformation. The density matrix of the quantum register is the sum of $M=2$ source systems, where source system $m=1$ identifies the valuable unknown input quantum system stored in the quantum register, while $m=2$ identifies an unknown undesired residual quantum system. The $U_{F} $ unitary evaluates $K$ bases for the source system and defines a $W$ auxiliary quantum system. The $U_{CQT} $ unitary is a preliminary operation for the partitioning of the $K$ bases onto $M$ clusters via unitary $U_{P} $. The $U_{P} $ unitary regroups the bases with respect to the $M=2$ source systems. The results are then processed by the $\tilde{U}_{{\rm DSTFT}}^{\dag } $ and $U_{DFT} $ unitaries to extract the source system $m=1$ on the output of the memory unit.} 
 \label{fig4}
 \end{center}
\end{figure*}
\end{center}

\subsection{Experimental Implementation}
An experimental implementation of an HRE quantum memory in a near-term quantum device \cite{qc1} can integrate standard photonics devices, optical cavities and other fundamental physical devices. The quantum operations can be realized via the framework of gate-model quantum computations of near-term quantum devices \cite{qc1,qc2,qc3,qc4,qc5,qc6}, such as superconducting units \cite{qc2}. The application of a HRE quantum memory in a quantum Internet setting \cite{puj3,ref2,refn7,ref1,ref3} can be implemented via noisy quantum links between the quantum repeaters \cite{ref6,refn6,ref50,ref51,ref53,ref56,ref58} (e.g., optical fibers \cite{ref28,ref47,ref5}, wireless quantum channels \cite{sat,telep}, free-space optical channels \cite{ref30}) and fundamental quantum transmission protocols \cite{ref4,add4,qkdrev,ref22}.

\section{Integrated Local Unitaries}
\label{sec3}
This section defines the local unitary operations integrated on an HRE quantum memory unit.

\subsection{Quantum Machine Learning Unitary}
The $U_{ML} $ quantum machine learning unitary implements an unsupervised learning for a blind separation of the unlabeled quantum register. The $U_{ML} $ unitary is defined as
\begin{equation} \label{20)} 
U_{ML} =U_{F} U_{CQT} U_{P} U_{CQT}^{\dag } ,   
\end{equation} 
where $U_{F} $ is a factorization unitary, $U_{CQT} $ is the quantum constant $Q$ transform, $U_{P} $ is a partitioning unitary, while $U_{CQT}^{\dag } $ is the inverse of $U_{CQT} $.

\subsubsection{Factorization Unitary}
\begin{theorem}
(Factorization of the unknown mixed quantum system of the quantum register). The $U_{F} $ unitary factorizes the unknown $\sigma _{QR} $ mixed quantum system of the $QR$ quantum register into a unitary $u_{mk} =e^{{-iH_{mk} \tau \mathord{\left/ {\vphantom {-iH_{mk} \tau  \hbar }} \right. \kern-\nulldelimiterspace} \hbar } } $, with a Hamiltonian $H_{mk} $ and application time $\tau $, and into a system $w_{kt} $, where $t=1,\ldots ,T$, $m=1,\ldots ,M$, and $k=1,\ldots ,K$, and where $T$ is the evolution time, $M$ is the number of source systems of $\sigma _{QR} $, and $K$ is the number of bases. 
\end{theorem}
\begin{proof}
The aim of the $U_{F} $ factorization unitary is to factorize the mixed quantum register \eqref{ZEqnNum669290} into a basis matrix $U_{B} $ and a quantum system $\vec{\rho }_{W} $, as
\begin{equation} \label{ZEqnNum506051} 
\begin{split}
   {{U}_{F}}{{{\vec{\sigma }}}_{QR}}U_{F}^{\dagger }&={{U}_{F}}\left( {{U}_{QR}}{{\rho }_{in}}U_{QR}^{\dagger } \right)U_{F}^{\dagger } \\ 
 & ={{U}_{B}}{{{\vec{\rho }}}_{W}}U_{B}^{\dagger },  
\end{split}
\end{equation} 
where $U_{B} $ is a complex basis matrix, defined as
\begin{equation} \label{ZEqnNum997483} 
U_{B} =\left\{u_{mk} \right\}\in {\mathbb{C}}^{M\times K} ,  
\end{equation} 
and $\vec{\rho }_{W} \in {\mathbb{C}}^{K\times T} $ is a complex matrix, defined as
\begin{equation} \label{23)} 
\vec{\rho }_{W} =\left\{\rho _{W}^{\left(t\right)} \right\}_{t=1}^{T} ,  
\end{equation} 
where
\begin{equation} \label{24)} 
\begin{split}
   \rho _{W}^{\left( t \right)}&=\sum\limits_{k=1}^{K}{v_{k}^{\left( t \right)}\left| {{\phi }_{k}} \right\rangle \left\langle  {{\phi }_{k}} \right|} \\ 
 & =\sum\limits_{k=1}^{K}{\sqrt{v_{k}^{\left( t \right)}}\left| {{\phi }_{k}} \right\rangle \sqrt{v_{k}^{\left( t \right)}}\left\langle  {{\phi }_{k}} \right|} \\ 
 & =\sum\limits_{k=1}^{K}{W_{k}^{\left( t \right)}{{\left( W_{k}^{\left( t \right)} \right)}^{\dagger }}},  
\end{split}
\end{equation} 
where $0\le v_{k}^{\left(t\right)} \le 1$, and $\sum _{k=1}^{K}v_{k}^{\left(t\right)}  =1$, while $K$ is the total number of bases of $U_{B} $, while $W_{k}^{\left(t\right)} \in {\mathbb{C}}$ is a complex quantity, as
\begin{equation} \label{ZEqnNum286942} 
W_{k}^{\left(t\right)} =\sqrt{v_{k}^{\left(t\right)} } {\left| \phi _{k}  \right\rangle} .  
\end{equation} 
The first part of the problem is therefore to find \eqref{ZEqnNum997483}, where $u_{mk} $ is a unitary that sets a computational basis for $W_{k}^{\left(t\right)} $ in \eqref{ZEqnNum286942}, defined as
\begin{equation} \label{26)} 
u_{mk} =e^{{-iH_{mk} \tau \mathord{\left/ {\vphantom {-iH_{mk} \tau  \hbar }} \right. \kern-\nulldelimiterspace} \hbar } } ,  
\end{equation} 
where $H_{mk} $ is a Hamiltonian, as
\begin{equation} \label{27)} 
H_{mk} =G_{mk} {\left| k_{m}  \right\rangle} {\left\langle k_{m}  \right|} ,  
\end{equation} 
where $G_{mk} $ is the eigenvalue of basis ${\left| k_{m}  \right\rangle} $, $H_{mk} {\left| k_{m}  \right\rangle} =G_{mk} {\left| k_{m}  \right\rangle} $, while $\tau $ is the application time of $u_{mk} $.

The second part of the problem is to determine $W$, as
\begin{equation} \label{ZEqnNum916898} 
W=\left\{W_{k}^{\left(t\right)} =w_{kt} \right\}\in {\mathbb{C}}^{K\times T} ,  
\end{equation} 
where $W_{k}^{\left(t\right)} =w_{kt} $ is a system state, that formulates $\tilde{X}^{\left(m,t\right)} $ as
\begin{equation} \label{ZEqnNum112860} 
\begin{split}
   {{{\tilde{X}}}^{\left( m,t \right)}}&={{\left[ {{U}_{B}}W \right]}_{mt}} \\ 
 & =\sum\limits_{k=1}^{K}{{{u}_{mk}}{{w}_{kt}}},  
\end{split}
\end{equation} 
where $\tilde{X}^{\left(m,t\right)} $ is an approximation of $X^{\left(m,t\right)} $, 
\begin{equation} \label{ZEqnNum918527} 
\tilde{X}^{\left(m,t\right)} \approx X^{\left(m,t\right)} ,  
\end{equation} 
where $X^{\left(m,t\right)} $ is defined in \eqref{ZEqnNum440486}.

As follows, for the total evolution time $T$, $\vec{X}\in {\mathbb{C}}^{M\times T} $ can be defined as
\begin{equation} \label{ZEqnNum956184} 
\vec{X}=\left\{X^{\left(1,t\right)} ,\ldots ,X^{\left(M,t\right)} \right\}_{t=1}^{T} ,   
\end{equation} 
and the challenge is to evaluate \eqref{ZEqnNum956184} as a decomposition
\begin{equation} \label{ZEqnNum852487} 
\begin{split}
   \tilde{X}&={{U}_{B}}W \\ 
 & ={{e}^{{-i{{H}_{\Sigma }}\tau }/{\hbar }\;}}W \\ 
 & =\sum\limits_{m=1}^{M}{\sum\limits_{t=1}^{T}{\sum\limits_{k=1}^{K}{{{u}_{mk}}W_{k}^{\left( t \right)}}}} \\ 
 & =\sum\limits_{m=1}^{M}{\sum\limits_{t=1}^{T}{\sum\limits_{k=1}^{K}{{{e}^{{-i{{H}_{mk}}\tau }/{\hbar }\;}}}\sqrt{v_{k}^{\left( t \right)}}\left| \phi _{k}^{\left( m,t \right)} \right\rangle .}}  
\end{split}
\end{equation} 
Thus, by applying of the $u_{mk} $ unitaries for the total evolution time $T$, $\tilde{X}\in {\mathbb{C}}^{M\times T} $ is as
\begin{equation} \label{ZEqnNum106419} 
\begin{split}
   \tilde{X}&={{U}_{B}}W \\ 
 & =\sum\limits_{m=1}^{M}{\sum\limits_{k=1}^{K}{{{\ell }_{m,k}}\left( \tau  \right)}\left| {{k}_{m}} \right\rangle ,} \\ 
 & =\alpha \left( \sum\limits_{{{k}_{1}}=1}^{{{K}_{1}}}{\left| {{k}_{1}} \right\rangle +\ldots +\sum\limits_{{{k}_{M}}=1}^{{{K}_{M}}}{\left| {{k}_{M}} \right\rangle }} \right),  
\end{split}
\end{equation} 
where $K_{m} $ is the number of bases associated with the $m$-th source system,
\begin{equation} \label{34)} 
\sum _{m=1}^{K}K_{m}  =K,  
\end{equation} 
and $0\le \left|\ell _{m,k} \left(\tau \right)\right|^{2} \le 1$, $\sum _{m=1}^{M}\sum _{k=1}^{K}\left|\ell _{m,k} \left(\tau \right)\right|^{2}   =1$.

In our setting $M=2$, and our aim is to get the system state $m=1$ on the output of the HRE quantum memory, thus a ${\left| \Phi ^{*}  \right\rangle} $ target output system state is defined as
\begin{equation} \label{ZEqnNum423447} 
{\left| \Phi ^{*}  \right\rangle} ={\textstyle\frac{1}{\sqrt{K_{1} } }} \sum _{k_{1} =1}^{K_{1} }{\left| k_{1}  \right\rangle}  ,  
\end{equation} 
where $K_{1} $ is the number of bases for source system $m=1$, $k_{1} =1,\ldots ,K_{1} $.

Let rewrite the system state $\tilde{X}$ \eqref{ZEqnNum852487} as
\begin{equation} \label{ZEqnNum442220} 
\tilde{X}=\left\{\tilde{X}^{\left(1,t\right)} ,\ldots ,\tilde{X}^{\left(M,t\right)} \right\}_{t=1}^{T} ,   
\end{equation} 
and let
\begin{equation} \label{ZEqnNum808422} 
X^{\left(t\right)} =\sum _{m=1}^{M}X^{\left(m,t\right)}  ,   
\end{equation} 
and
\begin{equation} \label{ZEqnNum252043} 
\tilde{X}^{\left(t\right)} =\sum _{m=1}^{M}\tilde{X}^{\left(m,t\right)}  .  
\end{equation} 
Then, let $\rho _{\vec{X}} $ be a density matrix associated with $\vec{X}$, defined as 
\begin{equation} \label{39)} 
\rho _{\vec{X}} =\sum _{m=1}^{M}\sum _{t=1}^{T}\vec{X}^{\left(m,t\right)} \left(\vec{X}^{\left(m,t\right)} \right)^{\dag }    
\end{equation} 
and let
\begin{equation} \label{40)} 
\rho _{\tilde{X}} =\sum _{m=1}^{M}\sum _{t=1}^{T}\tilde{X}^{\left(m,t\right)} \left(\tilde{X}^{\left(m,t\right)} \right)^{\dag }    
\end{equation} 
be the density matrix associated with \eqref{ZEqnNum442220}. 

The aim of the estimation is to minimize the $D\left(\left. \cdot \right\| \cdot \right)$ quantum relative entropy function taken between $\rho _{\vec{X}} $ and $\rho _{\tilde{X}} $, thus an $f\left(U_{F} \right)$ objective function for $U_{F} $ is defined via \eqref{ZEqnNum808422} and \eqref{ZEqnNum252043} as
\begin{equation} \label{ZEqnNum162412} 
\begin{split}
   f\left( {{U}_{F}} \right)&=\underset{{\tilde{X}}}{\mathop{\min }}\,D\left( \left. {{\rho }_{{\vec{X}}}} \right\|{{\rho }_{{\tilde{X}}}} \right) \\ 
 & =\underset{{\tilde{X}}}{\mathop{\min }}\,\text{Tr}\left( {{\rho }_{{\vec{X}}}}\log \left( {{\rho }_{{\vec{X}}}} \right) \right)-\text{Tr}\left( {{\rho }_{{\vec{X}}}}\log \left( {{\rho }_{{\tilde{X}}}} \right) \right).  
\end{split}
\end{equation} 
To achieve the objective function $f\left(U_{F} \right)$ in \eqref{ZEqnNum162412}, a factorization method is defined for $U_{F} $ that is based on the fundamentals of Bayesian nonnegative matrix factorization \cite{nmf1,nmf2,nmf3,bnmf1,bnmf2,bnmf3,ss1,ss2,ss3,ss4} (Footnote: The $U_{F} $ factorization unitary applied on the mixed state of the quantum register is analogous to a Poisson-Exponential Bayesian nonnegative matrix factorization \cite{ss1,ss2,ss3,ss4} process.). The method adopts the Poisson distribution as ${\rm {\mathcal L}}\left(\cdot \right)$ likelihood function and the exponential distribution for the control parameters \cite{ss1,ss2,ss3,ss4} $\alpha _{mk} $ and $\beta _{kt} $ defined for the controlling of $u_{mk} $ and $w_{kt} $. 

Let $u_{mk} $ and $w_{kt} $ from \eqref{ZEqnNum112860} be defined via the control parameters $\alpha _{mk} $ and $\beta _{kt} $ as exponential distributions
\begin{equation} \label{ZEqnNum127140} 
u_{mk} \simeq \alpha _{mk} e^{-\alpha _{mk} u_{mk} } ,  
\end{equation} 
with mean $\alpha _{mk}^{-1} $, and 
\begin{equation} \label{ZEqnNum324268} 
w_{kt} \simeq \beta _{kt} e^{-\beta _{kt} w_{kt} } ,  
\end{equation} 
with mean $\beta _{kt}^{-1} $.

Using \eqref{ZEqnNum162412}, \eqref{ZEqnNum127140} and \eqref{ZEqnNum324268}, a ${\rm {\mathcal L}}\left(\cdot \right)$ log likelihood function
\begin{equation} \label{44)} 
-{\rm {\mathcal L}}\left(x,\left. y\right|z\right)=-\log \Pr \left(x,\left. y\right|z\right) 
\end{equation} 
can be defined as
\begin{equation} \label{ZEqnNum431999}
\begin{split}
   -\mathcal{L}&\left( {{U}_{B}},W\left| {\vec{X}} \right. \right) \\ 
 & =D\left( \left. {{\rho }_{{\vec{X}}}} \right\|{{\rho }_{{\tilde{X}}}} \right)+\sum\limits_{m=1}^{M}{\sum\limits_{k=1}^{K}{{{\alpha }_{mk}}{{u}_{mk}}{{\left( {{\alpha }_{mk}}{{u}_{mk}} \right)}^{\dagger }}}}+\sum\limits_{k=1}^{K}{\sum\limits_{t=1}^{T}{{{\beta }_{kt}}{{w}_{kt}}{{\left( {{\beta }_{kt}}{{w}_{kt}} \right)}^{\dagger }}}},  
\end{split}
\end{equation}
thus the objective function $f\left(U_{F} \right)$ can be rewritten via as \eqref{ZEqnNum431999}
\begin{equation} \label{46)} 
f\left(U_{F} \right)=\mathop{\min }\limits_{\tilde{X}} \left(-{\rm {\mathcal L}}\left(U_{B} ,W\left|\vec{X}\right. \right)\right).  
\end{equation} 
The problem is therefore can be reduced to determine the model parameters 
\begin{equation} \label{ZEqnNum225569} 
\zeta =\left\{U_{B} ,W\right\} 
\end{equation} 
that are treated as latent variables for the estimation of the control parameters \cite{ss1,ss2,ss3,ss4,bnmf1,bnmf2,bnmf3}
\begin{equation} \label{ZEqnNum356288} 
\tau _{mk}^{\left(t\right)} =\left\{\alpha _{mk} ,\beta _{kt} \right\}.  
\end{equation} 
A maximum likelihood estimation $\tilde{\zeta }$ of \eqref{ZEqnNum225569} is as
\begin{equation} \label{ZEqnNum761868} 
\tilde{\zeta }=\arg \mathop{\max }\limits_{\zeta } {\rm {\mathcal D}}\left(\left. \vec{X}\right|\zeta \right),  
\end{equation} 
where ${\rm {\mathcal D}}\left(\cdot \right)$ is some distribution, that identifies an incomplete estimation problem.

The estimation of \eqref{ZEqnNum225569} can also be yielded from a maximization of a marginal likelihood function ${\rm {\mathcal L}}\left(\left. \vec{X}\right|\zeta \right)$ as
\begin{equation} \label{ZEqnNum374829} 
{\rm {\mathcal L}}\left(\left. \vec{X}\right|\zeta \right)={\int\limits\!\!\!\!\int}\sum _{\vec{\kappa }}{\rm {\mathcal D}}\left(\left. \vec{X}\right|\vec{\kappa }\right){\rm {\mathcal D}}\left(\left. \vec{\kappa }\right|U_{B} ,W\right){\rm {\mathcal D}}\left(U_{B} ,\left. W\right|\zeta \right)dU_{B} dW  , 
\end{equation} 
where $\vec{\kappa }$ is a complex matrix, $\vec{\kappa }\in {\mathbb{C}}^{M\times T} $,
\begin{equation} \label{51)} 
\vec{\kappa }=\left\{\kappa ^{\left(1,t\right)} ,\ldots ,\kappa ^{\left(M,t\right)} \right\}_{t=1}^{T} ,  
\end{equation} 
where
\begin{equation} \label{52)} 
\kappa ^{\left(m,t\right)} =\left(\kappa _{k=1}^{\left(m,t\right)} ,\ldots ,\kappa _{k=K}^{\left(m,t\right)} \right)^{T} , 
\end{equation} 
with 
\begin{equation} \label{53)} 
\kappa _{k}^{\left(m,t\right)} =\kappa _{mkt}  
\end{equation} 
where
\begin{equation} \label{ZEqnNum831401} 
\kappa _{mkt} =u_{mk} w_{kt} .  
\end{equation} 
The quantity in \eqref{ZEqnNum831401} can be estimated via \eqref{ZEqnNum127140} and \eqref{ZEqnNum324268} as
\begin{equation} \label{ZEqnNum972723} 
\kappa _{mkt} \approx \alpha _{mk} e^{-\alpha _{mk} u_{mk} } \beta _{kt} e^{-\beta _{kt} w_{kt} } .  
\end{equation} 
Using \eqref{ZEqnNum831401}, $\tilde{X}^{\left(m,t\right)} $ in \eqref{ZEqnNum112860} can be rewritten as
\begin{equation} \label{56)} 
\tilde{X}^{\left(m,t\right)} =\sum _{m=1}^{M}\sum _{k=1}^{K}\kappa _{mkt}   . 
\end{equation} 
However, since the exact solution does not exists \cite{ss1,ss2,ss3,ss4}, since it would require the factorization of ${\rm {\mathcal D}}\left(\left. \vec{\kappa },U_{B} ,W\right|\vec{X},\zeta \right)$, such that $\zeta ,U_{B} ,W$ are unknown. 

This problem can be solved by a variational Bayesian inference procedure \cite{ss1,ss2,ss3,ss4,bnmf1,bnmf2,bnmf3}, via the maximization of the lower bound of a likelihood function ${\rm {\mathcal L}}_{{\rm {\mathcal D}}_{v} } $
\begin{equation} \label{ZEqnNum672915}
\begin{split}
   {{\mathcal{L}}_{{{\mathcal{D}}_{v}}}}&=\iint\limits{\sum\limits_{{\vec{\kappa }}}{{{\mathcal{D}}_{v}}\left( \vec{\kappa },{{U}_{B}},W \right)\log \tfrac{\mathcal{D}\left( \vec{X},\vec{\kappa },{{U}_{B}},\left. W \right|\zeta  \right)}{{{\mathcal{D}}_{v}}\left( \vec{\kappa },{{U}_{B}},W \right)}d{{U}_{B}}dW}} \\ 
 & =\mathbb{E}\left( \log \mathcal{D}\left( \vec{X},\vec{\kappa },{{U}_{B}},\left. W \right|\zeta  \right) \right)+H\left( {{\mathcal{D}}_{v}}\left( \vec{\kappa },{{U}_{B}},W \right) \right),  
\end{split}
\end{equation}
where ${\rm {\mathcal D}}_{v} $ is a variational distribution, while $H\left({\rm {\mathcal D}}_{v} \left(\vec{\kappa },U_{B} ,W\right)\right)$ is the entropy of variational distribution ${\rm {\mathcal D}}_{v} \left(\vec{\kappa },U_{B} ,W\right)$, 
\begin{equation}\label{ZEqnNum532577}
H\left( {{\mathcal{D}}_{v}}\left( \vec{\kappa },{{U}_{B}},W \right) \right)=\sum\limits_{m=1}^{M}{\sum\limits_{t=1}^{T}{H\left( {{\kappa }^{\left( m,t \right)}} \right)}}+\sum\limits_{m=1}^{M}{\sum\limits_{k=1}^{K}{H\left( {{u}_{mk}} \right)}}+\sum\limits_{k=1}^{K}{\sum\limits_{t=1}^{T}{H\left( {{w}_{kt}} \right),}}
\end{equation}
and where ${\rm {\mathcal D}}_{v} \left(\vec{\kappa },U_{B} ,W\right)$ is a joint variational distribution, as
\begin{equation}\label{ZEqnNum150251}
\begin{split}
   {{\mathcal{D}}_{v}}\left( \vec{\kappa },{{U}_{B}},W \right)&={{\mathcal{D}}_{v}}\left( {\vec{\kappa }} \right){{\mathcal{D}}_{v}}\left( {{U}_{B}} \right){{\mathcal{D}}_{v}}\left( W \right) \\ 
 & =\prod\limits_{m}{\prod\limits_{t}{\prod\limits_{k}{{{\mathcal{D}}_{v}}\left( {{\kappa }_{mkt}} \right){{\mathcal{D}}_{v}}\left( {{u}_{mk}} \right){{\mathcal{D}}_{v}}\left( {{w}_{kt}} \right)}}},  
\end{split}
\end{equation}
from which distribution ${\rm {\mathcal D}}\left(\left. \vec{\kappa },U_{B} ,W\right|\vec{X},\zeta \right)$ can be approximated as \cite{ss1,ss2,ss3,ss4}
\begin{equation} \label{60)} 
{\rm {\mathcal D}}\left(\left. \vec{\kappa },U_{B} ,W\right|\vec{X},\zeta \right)\approx \prod _{m}\prod _{t}\prod _{k}{\rm {\mathcal D}}_{v} \left(\kappa _{mkt} \right){\rm {\mathcal D}}_{v} \left(u_{mk} \right){\rm {\mathcal D}}_{v} \left(w_{kt} \right)   .  
\end{equation} 
The function ${\rm {\mathcal L}}_{{\rm {\mathcal D}}_{v} } $ in \eqref{ZEqnNum672915} is related to \eqref{ZEqnNum374829} as
\begin{equation} \label{61)} 
{\rm {\mathcal L}}\left(\left. \vec{X}\right|\zeta \right)\ge {\rm {\mathcal L}}_{{\rm {\mathcal D}}_{v} } .  
\end{equation} 
The result in \eqref{ZEqnNum150251} therefore also determines the number $K$ of bases selected for the factorization unitary $U_{F} $. The ${\rm {\mathcal D}}_{v} $ variational distributions ${\rm {\mathcal D}}_{v} \left(\kappa _{mkt} \right)$, ${\rm {\mathcal D}}_{v} \left(u_{k} \right)$ and ${\rm {\mathcal D}}_{v} \left(w_{kt} \right)$ are determined for the unitary $U_{F} $ as follows.

Let ${\rm {\mathcal D}}_{v} \left(\Phi \right)$ refer to the variational distribution of a given $\Phi $,
\begin{equation} \label{ZEqnNum956463} 
\Phi \in \left\{\vec{\kappa },U_{B} ,W\right\}.   
\end{equation} 
Since only the joint (posterior) distribution ${\rm {\mathcal D}}\left(\left. \vec{X},\vec{\kappa },U_{B} ,W\right|\zeta \right)$ is obtainable, the variational distributions have to be evaluated as
\begin{equation} \label{63)} 
{\mathbb{E}}_{{\rm {\mathcal D}}_{v} \left(i\ne \Phi \right)} \left(\log {\rm {\mathcal D}}\left(\left. \vec{X},\vec{\kappa },U_{B} ,W\right|\zeta \right)\right)=\log {\rm {\mathcal D}}_{v} \left(\Phi \right),   
\end{equation} 
where ${\mathbb{E}}_{{\rm {\mathcal D}}_{v} \left(i\ne \Phi \right)} \left(\cdot \right)$ is the expectation function of the ${\rm {\mathcal D}}_{v} \left(i\right)$ variational distribution of $i$, such that $i\ne \Phi $, where $\Phi $ is as in \eqref{ZEqnNum956463}, with 
\begin{equation} \label{64)} 
{\mathbb{E}}_{a} \left(f\left(a\right)+g\left(a\right)\right)={\mathbb{E}}_{a} \left(f\left(a\right)\right)+{\mathbb{E}}_{a} \left(g\left(a\right)\right),  
\end{equation} 
for some functions $f\left(a\right)$ and $g\left(a\right)$, and
\begin{equation} \label{65)} 
{\mathbb{E}}_{a} \left(bf\left(a\right)\right)=b{\mathbb{E}}_{a} \left(f\left(a\right)\right) 
\end{equation} 
for some constant $b$, (note: for simplicity, we use ${\mathbb{E}}\left(\cdot \right)$ for the expectation function), while  
\begin{equation} \label{ZEqnNum213495}
\begin{split}
  & \log \mathcal{D}\left( \left. \vec{X},\vec{\kappa },{{U}_{B}},W \right|\zeta  \right) \\ 
  =&\sum\limits_{m=1}^{M}{\sum\limits_{t=1}^{T}{\log {{f}_{\delta }}\left( {{X}^{\left( m,t \right)}}-\sum\limits_{k=1}^{K}{{{\kappa }_{mkt}}} \right)}}+\sum\limits_{m=1}^{M}{\sum\limits_{k=1}^{K}{\sum\limits_{t=1}^{T}{\left( {{\kappa }_{mkt}}\log \left( {{u}_{mk}}{{w}_{kt}} \right) \right.}}} \\ 
 & -{{u}_{mk}}{{w}_{kt}}-\log {{f}_{\Gamma }}\left. \left( {{\kappa }_{mkt}}+1 \right) \right)+\sum\limits_{m=1}^{M}{\sum\limits_{k=1}^{K}{\left( \log {{\alpha }_{mk}}-{{\alpha }_{mk}}{{u}_{mk}} \right)}} \\ 
 & +\sum\limits_{k=1}^{K}{\sum\limits_{t=1}^{T}{\left( \log {{\beta }_{kt}}-{{\beta }_{kt}}{{w}_{kt}} \right),}}  
\end{split}
\end{equation}
where $f_{\delta } \left(\cdot \right)$ is the Dirac delta function, while $f_{\Gamma } \left(\cdot \right)$ is the Gamma function,
\begin{equation} \label{ZEqnNum865208} 
f_{\Gamma } \left(x\right)=\int\limits_{0}^{\infty }t^{x-1} e^{-t} dt .   
\end{equation} 
By utilizing a variational Poisson--Exponential Bayesian learning \cite{ss1,ss2,ss3,ss4}, these variational distributions can be evaluated as follows. 

The ${\rm {\mathcal D}}_{v} \left(\kappa _{mkt} \right)$ variational distribution is as
\begin{equation} \label{ZEqnNum489855} 
{\rm {\mathcal D}}_{v} \left(\kappa _{mkt} \right)={\rm {\mathcal M}}\left(\left. \kappa _{mkt} \right|\eta _{mkt} \right) 
\end{equation} 
where ${\rm {\mathcal M}}$ is a multinomial distribution, while $\eta _{mkt} $ is a multinomial parameter 
\begin{equation} \label{ZEqnNum425452} 
\eta _{mkt} =\frac{e^{{\mathbb{E}}\left(\log u_{mk} \right)+{\mathbb{E}}\left(\log w_{kt} \right)} }{\sum _{j}e^{{\mathbb{E}}\left(\log u_{mj} \right)+{\mathbb{E}}\left(\log w_{jt} \right)}  } ,   
\end{equation} 
while the ${\rm {\mathcal D}}_{v} \left(\kappa ^{\left(m,t\right)} \right)$ variational distribution is as
\begin{equation} \label{70)}
\begin{split}
  & \mathcal{M}\left( \left. {{\kappa }^{\left( m,t \right)}} \right|{{X}^{\left( m,t \right)}},\eta _{k}^{\left( m,t \right)} \right) \\ 
 & ={{f}_{\delta }}\left( {{X}^{\left( m,t \right)}}-\sum\limits_{k=1}^{K}{{{\kappa }_{mkt}}} \right){{X}^{\left( m,t \right)}}!\prod\limits_{k}{\tfrac{{{\left( {{\eta }_{mkt}} \right)}^{{{\kappa }_{mkt}}}}}{{{\kappa }_{mkt}}!}},  
\end{split}
\end{equation} 
where $\eta _{k}^{\left(m,t\right)} $ is a multinomial parameter vector
\begin{equation} \label{71)} 
\eta _{k}^{\left(m,t\right)} =\left(\eta _{k=1}^{\left(m,t\right)} ,\ldots ,\eta _{k=K}^{\left(m,t\right)} \right)^{T} ,  
\end{equation} 
such that
\begin{equation} \label{72)} 
\sum _{k=1}^{K}\eta _{k}^{\left(m,t\right)}  =1.  
\end{equation} 
The ${\rm {\mathcal D}}_{v} \left(u_{mk} \right)$ variational distribution is as
\begin{equation} \label{ZEqnNum204512}
\begin{split}
  & {{\mathcal{D}}_{v}}\left( {{u}_{mk}} \right) \\ 
 & ={{e}^{\left( \sum\limits_{t=1}^{T}{\mathbb{E}\left( {{\kappa }_{mkt}} \right)}\log {{u}_{mk}}-\left( \sum\limits_{t=1}^{T}{\mathbb{E}\left( {{w}_{kt}} \right)}+{{\alpha }_{mk}} \right){{u}_{mk}} \right)}} \\ 
 & =\mathcal{G}\left( \left. {{u}_{mk}} \right|{{{\tilde{\alpha }}}_{mk}}\left( A \right),{{{\tilde{\alpha }}}_{mk}}\left( B \right) \right),  
\end{split}
\end{equation} 
where ${\rm {\mathcal G}}\left(\cdot \right)$ is a Gamma distribution, 
\begin{equation} \label{ZEqnNum162591} 
{\rm {\mathcal G}}\left(x;a,b\right)=e^{\left(a-1\right)\log x-\frac{x}{b} -\log f_{\Gamma } \left(a\right)-a\log b} , 
\end{equation} 
where $a$ is a shape parameter, while $b$ is a scale parameter, $f_{\Gamma } \left(\cdot \right)$ is the Gamma function \eqref{ZEqnNum865208}. The entropy of \eqref{ZEqnNum162591} is as
\begin{equation} \label{75)} 
H\left({\rm {\mathcal G}}\left(x;a,b\right)\right)=-\left(a-1\right)\partial _{{\rm {\mathcal G}}_{\log } } \left(a\right)+\log b+a+\log f_{\Gamma } \left(a\right), 
\end{equation} 
where $\partial _{{\rm {\mathcal G}}_{\log } } \left(\cdot \right)$ is the derivative of the log gamma function (digamma function), 
\begin{equation} \label{ZEqnNum655695} 
\partial _{{\rm {\mathcal G}}_{\log } } \left(x\right)={\textstyle\frac{d\log f_{\Gamma } \left(x\right)}{dx}} ,  
\end{equation} 
while ${\mathbb{E}}\left(\kappa _{mkt} \right)$ is evaluated as
\begin{equation} \label{77)} 
{\mathbb{E}}\left(\kappa _{mkt} \right)=X^{\left(m,t\right)} \eta _{mkt} ,  
\end{equation} 
while $\tilde{\alpha }_{mk} \left(A\right)$ and $\tilde{\alpha }_{mk} \left(B\right)$ are control parameters for $U_{B} $, defined as 
\begin{equation} \label{ZEqnNum525301} 
\tilde{\alpha }_{mk} \left(A\right)=1+\sum _{t=1}^{T}{\mathbb{E}}\left(\kappa _{mkt} \right) , 
\end{equation} 
while $\tilde{\alpha }_{mk} \left(B\right)$ is defined as
\begin{equation} \label{ZEqnNum212545} 
\tilde{\alpha }_{mk} \left(B\right)={\textstyle\frac{1}{\sum _{t=1}^{T}{\mathbb{E}}\left(w_{kt} \right) +\alpha _{mk} }} .  
\end{equation} 
The ${\rm {\mathcal D}}_{v} \left(w_{kt} \right)$ variational distribution is as
\begin{equation}\label{ZEqnNum334896}
\begin{split}
  & {{\mathcal{D}}_{v}}\left( {{w}_{kt}} \right) \\ 
 & ={{e}^{\left( \sum\limits_{m=1}^{M}{\mathbb{E}\left( {{\kappa }_{mkt}} \right)}\log {{w}_{kn}}-\left( \sum\limits_{m=1}^{M}{\mathbb{E}\left( {{u}_{mk}} \right)}+{{\beta }_{kt}} \right){{w}_{kt}} \right)}} \\ 
 & =\mathcal{G}\left( \left. {{w}_{kt}} \right|{{{\tilde{\beta }}}_{kt}}\left( A \right),{{{\tilde{\beta }}}_{kt}}\left( B \right) \right), 
\end{split}
\end{equation}
where $\tilde{\beta }_{kt} \left(A\right)$ and $\tilde{\beta }_{kt} \left(B\right)$ are control parameters for $W$, defined as
\begin{equation} \label{ZEqnNum494117} 
\tilde{\beta }_{kt} \left(A\right)=1+\sum _{m=1}^{M}{\mathbb{E}}\left(\kappa _{mkt} \right) , 
\end{equation} 
and
\begin{equation} \label{ZEqnNum334111} 
\tilde{\beta }_{kt} \left(B\right)={\textstyle\frac{1}{\sum _{m=1}^{M}{\mathbb{E}}\left(u_{mk} \right) +\beta _{kt} }} .  
\end{equation} 
Given the variational parameters $\tilde{\alpha }_{mk} \left(A\right)$, $\tilde{\alpha }_{mk} \left(B\right)$, $\tilde{\beta }_{kt} \left(A\right)$ and $\tilde{\beta }_{kt} \left(B\right)$ in \eqref{ZEqnNum525301}, \eqref{ZEqnNum212545}, \eqref{ZEqnNum494117} and \eqref{ZEqnNum334111}, the estimates of $U_{B} $ and $W$ are realized by the determination of the Gamma means ${\mathbb{E}}\left(u_{mk} \right)$ and ${\mathbb{E}}\left(w_{kt} \right)$ \cite{ss1,ss2,ss3,ss4}. It can be verified that the mean ${\mathbb{E}}\left(w_{kt} \right)$ in \eqref{ZEqnNum204512}, \eqref{ZEqnNum212545} and \eqref{ZEqnNum334896} can be evaluated via \eqref{ZEqnNum494117} and \eqref{ZEqnNum334111} as a mean of a Gamma distribution
\begin{equation} \label{83)} 
{\mathbb{E}}\left(w_{kt} \right)=\tilde{\beta }_{kt} \left(A\right)\tilde{\beta }_{kt} \left(B\right),  
\end{equation} 
while ${\mathbb{E}}\left(\log w_{kt} \right)$ is as
\begin{equation} \label{84)} 
{\mathbb{E}}\left(\log w_{kt} \right)=\partial _{{\rm {\mathcal G}}_{\log } } \left(\tilde{\beta }_{kt} \left(A\right)\right)+\log \tilde{\beta }_{kt} \left(B\right),  
\end{equation} 
where $\partial _{{\rm {\mathcal G}}_{\log } } \left(\cdot \right)$ digamma function \eqref{ZEqnNum655695}. 

The mean ${\mathbb{E}}\left(u_{mk} \right)$ in \eqref{ZEqnNum334896} and \eqref{ZEqnNum334111} can be evaluated via \eqref{ZEqnNum525301} and \eqref{ZEqnNum212545}, as a mean of a Gamma distribution
\begin{equation} \label{85)} 
{\mathbb{E}}\left(u_{mk} \right)=\tilde{\alpha }_{mk} \left(A\right)\tilde{\alpha }_{mk} \left(B\right),  
\end{equation} 
and ${\mathbb{E}}\left(\log u_{mk} \right)$ is yielded as
\begin{equation} \label{86)} 
{\mathbb{E}}\left(\log u_{mk} \right)=\partial _{{\rm {\mathcal G}}_{\log } } \left(\tilde{\alpha }_{mk} \left(A\right)\right)+\log \tilde{\alpha }_{mk} \left(B\right).  
\end{equation} 
As the ${\rm {\mathcal D}}_{v} \left(\kappa _{mkt} \right)$, ${\rm {\mathcal D}}_{v} \left(u_{mk} \right)$ and ${\rm {\mathcal D}}_{v} \left(w_{kt} \right)$ variational distributions are determined via \eqref{ZEqnNum489855}, \eqref{ZEqnNum204512} and \eqref{ZEqnNum334896}the evaluation of \eqref{ZEqnNum150251} is straightforward.

Using the defined terms, the term ${\mathbb{E}}\left(\log {\rm {\mathcal D}}\left(\vec{X},\vec{\kappa },U_{B} ,\left. W\right|\zeta \right)\right)$ from \eqref{ZEqnNum672915} can be evaluated as
\begin{equation} \label{ZEqnNum518287}
\begin{split}
  & \mathbb{E}\left( \log \mathcal{D}\left( \vec{X},\vec{\kappa },{{U}_{B}},\left. W \right|\zeta  \right) \right) \\ 
  =&\sum\limits_{m=1}^{M}{\sum\limits_{t=1}^{T}{\mathbb{E}\left( \log {{f}_{\delta }}\left( {{X}^{\left( m,t \right)}}-\sum\limits_{k=1}^{K}{{{\kappa }_{mkt}}} \right) \right)}}+\sum\limits_{m=1}^{M}{\sum\limits_{k=1}^{K}{\mathbb{E}\left( \log {{u}_{mk}} \right)\sum\limits_{t=1}^{T}{\mathbb{E}\left( {{\kappa }_{mkt}} \right)}}} \\ 
 & +\sum\limits_{k=1}^{K}{\sum\limits_{t=1}^{T}{\mathbb{E}\left( \log {{w}_{kt}} \right)\sum\limits_{m=1}^{M}{\mathbb{E}\left( {{\kappa }_{mkt}} \right)}-}}\sum\limits_{m=1}^{M}{\sum\limits_{k=1}^{K}{\sum\limits_{t=1}^{T}{\mathbb{E}\left( {{u}_{mk}} \right)\mathbb{E}\left( {{w}_{kt}} \right)}}} \\ 
 & -\sum\limits_{m=1}^{M}{\sum\limits_{k=1}^{K}{\sum\limits_{t=1}^{T}{\mathbb{E}\left( \log {{f}_{\Gamma }}\left( {{\kappa }_{mkt}}+1 \right) \right)}}} \\ 
 & +\sum\limits_{m=1}^{M}{\sum\limits_{k=1}^{K}{\left( \log {{\alpha }_{mk}}-{{\alpha }_{mk}}\mathbb{E}\left( {{u}_{mk}} \right) \right)}} \\ 
 & +\sum\limits_{k=1}^{K}{\sum\limits_{t=1}^{T}{\left( \log {{\beta }_{kt}}-{{\beta }_{kt}}\mathbb{E}\left( {{w}_{kt}} \right) \right),}}  
\end{split}
\end{equation}
while the $H\left({\rm {\mathcal D}}_{v} \left(\vec{\kappa },U_{B} ,W\right)\right)$ entropy of the variational distribution from \eqref{ZEqnNum532577} can be evaluated as
\begin{equation} \label{ZEqnNum994815}
\begin{split}
  & H\left( {{\mathcal{D}}_{v}}\left( \vec{\kappa },{{U}_{B}},W \right) \right) \\ 
  =&\sum\limits_{m=1}^{M}{\sum\limits_{t=1}^{T}{\left( -\log {{f}_{\Gamma }}\left( {{X}^{\left( m,t \right)}}+1 \right)-\sum\limits_{k=1}^{K}{\mathbb{E}\left( {{\kappa }_{mkt}} \right)\log {{\eta }_{mkt}}} \right)}} \\ 
 & +\sum\limits_{m=1}^{M}{\sum\limits_{k=1}^{K}{\sum\limits_{t=1}^{T}{\mathbb{E}\left( \log {{f}_{\Gamma }}\left( {{\kappa }_{mkt}}+1 \right) \right)}}} \\ 
 & -\sum\limits_{m=1}^{M}{\sum\limits_{t=1}^{T}{\mathbb{E}\left( \log {{f}_{\delta }}\left( {{X}^{\left( m,t \right)}}-\sum\limits_{k=1}^{K}{{{\kappa }_{mkt}}} \right) \right)}} \\ 
 & +\sum\limits_{m=1}^{M}{\sum\limits_{k=1}^{K}{\left( -\left( {{{\tilde{\alpha }}}_{mk}}\left( A \right)-1 \right){{\partial }_{{{\mathcal{G}}_{\log }}}}\left( {{{\tilde{\alpha }}}_{mk}}\left( A \right) \right)+\log \left( {{{\tilde{\alpha }}}_{mk}}\left( B \right) \right)+{{{\tilde{\alpha }}}_{mk}}\left( A \right)+\log {{f}_{\Gamma }}\left( {{{\tilde{\alpha }}}_{mk}}\left( A \right) \right) \right)}} \\ 
 & +\sum\limits_{k=1}^{K}{\sum\limits_{t=1}^{T}{\left( -\left( {{{\tilde{\beta }}}_{kt}}\left( A \right)-1 \right){{\partial }_{{{\mathcal{G}}_{\log }}}}\left( {{{\tilde{\beta }}}_{k}}\left( A \right) \right)+\log \left( {{{\tilde{\beta }}}_{k}}\left( B \right) \right)+{{{\tilde{\beta }}}_{k}}\left( A \right)+\log {{f}_{\Gamma }}\left( {{{\tilde{\beta }}}_{k}}\left( A \right) \right) \right)}}.  
\end{split}
\end{equation}
Thus, from \eqref{ZEqnNum518287} and \eqref{ZEqnNum994815}, the lower bound ${\rm {\mathcal L}}_{{\rm {\mathcal D}}_{v} } $in \eqref{ZEqnNum672915} is as
\begin{equation} \label{ZEqnNum960383}
\begin{split}
  &{{\mathcal{L}}_{{{\mathcal{D}}_{v}}}} \\
   =&-\sum\limits_{m=1}^{M}{\sum\limits_{t=1}^{T}{\sum\limits_{k=1}^{K}{\mathbb{E}\left( {{u}_{mk}} \right)\mathbb{E}\left( {{w}_{kt}} \right)}}} \\ 
 & +\sum\limits_{m=1}^{M}{\sum\limits_{t=1}^{T}{\left( -\log {{f}_{\Gamma }}\left( {{X}^{\left( m,t \right)}}+1 \right)-\sum\limits_{k=1}^{K}{\mathbb{E}\left( {{\kappa }_{mkt}} \right)\log {{\eta }_{mkt}}} \right)}} \\ 
 & +\sum\limits_{m=1}^{M}{\sum\limits_{k=1}^{K}{\mathbb{E}\left( \log {{u}_{mk}} \right)\sum\limits_{t=1}^{T}{\mathbb{E}\left( {{\kappa }_{mkt}} \right)}}} \\ 
 & +\sum\limits_{k=1}^{K}{\sum\limits_{t=1}^{T}{\mathbb{E}\left( \log {{w}_{kt}} \right)\sum\limits_{m=1}^{M}{{{\kappa }_{mkt}}}}} \\ 
 & +\sum\limits_{m=1}^{M}{\sum\limits_{k=1}^{K}{\left( \log {{\alpha }_{mk}}-{{\alpha }_{mk}}\mathbb{E}\left( {{u}_{mk}} \right) \right)+}}\sum\limits_{k=1}^{K}{\sum\limits_{t=1}^{T}{\left( \log {{\beta }_{kt}}-{{\beta }_{kt}}\mathbb{E}\left( {{w}_{kt}} \right) \right)}} \\ 
 & +\sum\limits_{m=1}^{M}{\sum\limits_{k=1}^{K}{\left( -\left( {{{\tilde{\alpha }}}_{mk}}\left( A \right)-1 \right){{\partial }_{{{\mathcal{G}}_{\log }}}}\left( {{{\tilde{\alpha }}}_{mk}}\left( A \right) \right)+\log {{{\tilde{\alpha }}}_{mk}}\left( B \right)+{{{\tilde{\alpha }}}_{mk}}\left( A \right)+\log {{f}_{\Gamma }}\left( {{{\tilde{\alpha }}}_{mk}}\left( A \right) \right) \right)}} \\ 
 & +\sum\limits_{k=1}^{K}{\sum\limits_{t=1}^{T}{\left( -\left( {{{\tilde{\beta }}}_{kt}}\left( A \right)-1 \right){{\partial }_{{{\mathcal{G}}_{\log }}}}\left( {{{\tilde{\beta }}}_{kt}}\left( A \right) \right)+\log {{{\tilde{\beta }}}_{kt}}\left( B \right)+{{{\tilde{\beta }}}_{kt}}\left( A \right)+\log {{f}_{\Gamma }}\left( {{{\tilde{\beta }}}_{kt}}\left( A \right) \right) \right).}}  
\end{split}
\end{equation}
The next problem is the $\tilde{\tau }_{k}^{\left(t\right)} $ estimation of the control parameters $\alpha _{mk} ,\beta _{kt} $ in \eqref{ZEqnNum356288} as 
\begin{equation} \label{ZEqnNum645426} 
\tilde{\tau }_{mk}^{\left(t\right)} =\left\{E_{mk} ,F_{kt} \right\},  
\end{equation} 
such that $E_{mk} $ is a basis estimation
\begin{equation} \label{91)} 
E_{mk} \approx \alpha _{mk}  
\end{equation} 
and $F_{kt} $ is a system estimation
\begin{equation} \label{92)} 
F_{kt} \approx \beta _{kt} ,  
\end{equation} 
such that the variational lower bound ${\rm {\mathcal L}}_{{\rm {\mathcal D}}_{v} } $ in \eqref{ZEqnNum960383} is maximized \cite{ss1,ss2,ss3,ss4}. It is achieved for the unitary $U_{F} $ as follows. The maximization problem can be formalized via the $\partial \left({\rm {\mathcal L}}_{{\rm {\mathcal D}}_{v} } \right)$ derivative of ${\rm {\mathcal L}}_{{\rm {\mathcal D}}_{v} } $
\begin{equation} \label{93)} 
{\textstyle\frac{\partial \left({\rm {\mathcal L}}_{{\rm {\mathcal D}}_{v} } \right)}{\partial \alpha _{mk} }} ={\textstyle\frac{1}{\alpha _{mk} }} -{\mathbb{E}}\left(u_{mk} \right)+{\textstyle\frac{\partial \left(\log \left(\tilde{\alpha }_{mk} \left(B\right)\right)\right)}{\partial \alpha _{mk} }} =0,  
\end{equation} 
and
\begin{equation} \label{94)} 
{\textstyle\frac{\partial \left({\rm {\mathcal L}}_{{\rm {\mathcal D}}_{v} } \right)}{\partial \beta _{kt} }} ={\textstyle\frac{1}{\beta _{kt} }} -{\mathbb{E}}\left(w_{kt} \right)+{\textstyle\frac{\partial \left(\log \left(\tilde{\beta }_{kt} \left(B\right)\right)\right)}{\partial \beta _{kt} }} =0,  
\end{equation} 
which is solvable via \cite{ss1,ss3}
\begin{equation} \label{95)} 
\left(\alpha _{mk} \right)^{2} +\sum _{t=1}^{T}{\mathbb{E}}\left(w_{kt} \right)\alpha _{mk} - {\textstyle\frac{\sum _{t=1}^{T}{\mathbb{E}}\left(w_{kt} \right) }{{\mathbb{E}}\left(u_{mk} \right)}} =0,  
\end{equation} 
and
\begin{equation} \label{96)} 
\left(\beta _{kt} \right)^{2} +\sum _{m=1}^{M}{\mathbb{E}}\left(u_{mk} \right)\beta _{kt}  -{\textstyle\frac{\sum _{m=1}^{M}{\mathbb{E}}\left(u_{mk} \right) }{{\mathbb{E}}\left(w_{kt} \right)}} =0. 
\end{equation} 
After some calculations, $E_{mk} $ and $F_{kt} $ from \eqref{ZEqnNum645426} are as
\begin{equation} \label{ZEqnNum980009} 
E_{mk} ={\textstyle\frac{1}{2}} \left(-\sum _{t=1}^{T}{\mathbb{E}}\left(w_{kt} \right)+\left(\left(\sum _{t=1}^{T}{\mathbb{E}}\left(w_{kt} \right) \right)^{2} +4{\textstyle\frac{\sum _{t=1}^{T}{\mathbb{E}}\left(w_{kt} \right) }{{\mathbb{E}}\left(u_{mk} \right)}} \right)^{{\textstyle\frac{1}{2}} }  \right),  
\end{equation} 
and
\begin{equation} \label{ZEqnNum559004} 
F_{kt} ={\textstyle\frac{1}{2}} \left(-\sum _{m=1}^{M}{\mathbb{E}}\left(u_{mk} \right) +\left(\left(\sum _{m=1}^{M}{\mathbb{E}}\left(u_{mk} \right) \right)^{2} +4{\textstyle\frac{\sum _{m=1}^{M}{\mathbb{E}}\left(u_{mk} \right) }{{\mathbb{E}}\left(w_{kt} \right)}} \right)^{{\textstyle\frac{1}{2}} } \right), 
\end{equation} 
respectively.

From \eqref{ZEqnNum980009} and \eqref{ZEqnNum559004}, the $\tilde{\tau }_{mk}^{\left(t\right)} $ estimation in \eqref{ZEqnNum645426} is therefore straightforwardly yielded. Therefore, using the parameters $\tilde{\alpha }_{mk} \left(B\right),\tilde{\alpha }_{mk} \left(A\right),\tilde{\beta }_{kt} \left(A\right),\tilde{\beta }_{kt} \left(B\right)$ and $\eta _{mkt} $, the optimal variational distributions ${\rm {\mathcal D}}_{v} \left(\kappa _{mkt} \right)$, ${\rm {\mathcal D}}_{v} \left(u_{mk} \right)$ and ${\rm {\mathcal D}}_{v} \left(w_{kt} \right)$ can be substituted to estimate $\tilde{\tau }_{mk}^{\left(t\right)} $. 

Using \eqref{ZEqnNum980009} and \eqref{ZEqnNum559004}, the estimation of terms $u_{k} $ \eqref{ZEqnNum127140}, $w_{kt} $ \eqref{ZEqnNum324268} and $\kappa _{kt} $ \eqref{ZEqnNum972723} are yielded as
\begin{equation} \label{ZEqnNum908888} 
\tilde{u}_{mk} =E_{mk} e^{-E_{mk} \tilde{u}_{mk} } ,  
\end{equation} 
\begin{equation} \label{ZEqnNum508289} 
\tilde{w}_{kt} =F_{kt} e^{-F_{kt} \tilde{w}_{kt} } ,  
\end{equation} 
and
\begin{equation} \label{101)} 
\tilde{\kappa }_{mkt} =E_{mk} e^{-E_{mk} \tilde{u}_{mk} } F_{kt} e^{-F_{kt} \tilde{w}_{kt} } .  
\end{equation} 
The evaluation of \eqref{ZEqnNum980009} and \eqref{ZEqnNum559004} therefore is yielded in an iterative manner through the $\tilde{\alpha }_{mk} \left(B\right)$, $\tilde{\alpha }_{mk} \left(A\right)$, $\tilde{\beta }_{kt} \left(A\right)$, $\tilde{\beta }_{kt} \left(B\right)$ and $\eta _{mkt} $, and the $K^{*} $ optimal number of bases, $K$, is determined with respect to \eqref{ZEqnNum960383} such that
\begin{equation} \label{102)} 
K^{*} =\arg \mathop{\max }\limits_{K} {\rm {\mathcal L}}_{{\rm {\mathcal D}}_{v} } \left(K\right),  
\end{equation} 
where ${\rm {\mathcal L}}_{{\rm {\mathcal D}}_{v} } \left(K\right)$ refers to ${\rm {\mathcal L}}_{{\rm {\mathcal D}}_{v} } $ from \eqref{ZEqnNum960383} at a particular base number $K$. 

The proof is concluded here.
\end{proof}

The schematic representation of unitary $U_{F} $ is depicted in \fref{fig2}. 

 \begin{center}
\begin{figure*}[!htbp]
\begin{center}
\includegraphics[angle = 0,width=1\linewidth]{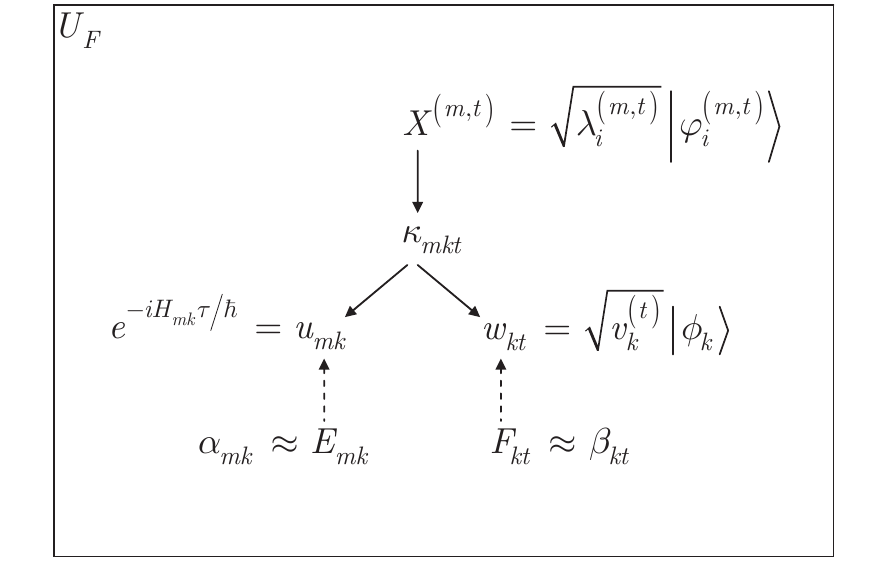}
\caption{Representation of the $U_{F} $ unitary over a total evolution time $T$, with $K$ factored bases and $M$ source systems ($M=2$ in our setting). The factorization is represented by the solid-line arrows. At a given $t$, $t=1,\ldots ,T$, the input system of $U_{F} $ subject of factorization is $X^{\left(m,t\right)} =\sqrt{\lambda _{i}^{\left(m,t\right)} } {\left| \varphi _{i}^{\left(m,t\right)}  \right\rangle} $, $m=1,\ldots ,M$. Term $\kappa _{mkt} $ is expressed as $\kappa _{mkt} =u_{mk} w_{kt} $, where ${{u}_{mk}}={{e}^{{-i{{H}_{mk}}\tau }/{\hbar }\;}}$ is a unitary, $u_{mk} \in {\mathbb{C}}$, $k=1,\ldots ,K$, which sets a computational basis for $w_{kt} $, $w_{kt} =W_{k}^{\left(t\right)} =\sqrt{v_{k}^{\left(t\right)} } {\left| \phi _{k}  \right\rangle} $. The basis matrix is $U_{B} =\left\{u_{mk} \right\}\in {\mathbb{C}}^{M\times K} $ with $K$ bases, $H_{mk} =G_{mk} {\left| k_{m}  \right\rangle} {\left\langle k_{m}  \right|} $ is a Hamiltonian, and $W=\left\{W_{k}^{\left(t\right)} =w_{kt} \right\}\in {\mathbb{C}}^{K\times T} $, $w_{kt} \in {\mathbb{C}}$. The factorization decomposes $X^{\left(m,t\right)} $ into $X^{\left(m,t\right)} =\left[U_{B} \vec{W}\right]_{mt} $, and for the total evolution $\vec{X}=U_{B} \vec{W}$, where $\vec{X}=\left\{X^{\left(1,t\right)} ,\ldots ,X^{\left(M,t\right)} \right\}_{t=1}^{T} $, while $\kappa _{kt} $ is as $\kappa _{mkt} =u_{mk} w_{kt} $. Terms $\alpha _{mk} $ and $\beta _{kt} $ are control parameters for $u_{mk} $ and $w_{kt} $ (controlling is depicted by the dashed-line arrows) to evaluate the parameters as $u_{mk} \simeq \alpha _{mk} e^{-\alpha _{mk} u_{mk} } $ and $w_{kt} \simeq \beta _{kt} e^{-\beta _{kt} w_{kt} } $, estimated by $E_{mk} $ and $F_{kt} $ as $\tilde{u}_{mk} =E_{mk} e^{-E_{mk} \tilde{u}_{mk} } $ and $\tilde{w}_{kt} =F_{kt} e^{-F_{kt} \tilde{w}_{kt} } $.} 
 \label{fig2}
 \end{center}
\end{figure*}
\end{center}

\subsubsection{Quantum Constant Q Transform}

As the $\left\{\tilde{u}_{mk} \right\}$ basis estimations \eqref{ZEqnNum908888} are determined via $\left\{E_{mk} \right\}$ \eqref{ZEqnNum980009}, the next problem is the partitioning of the $K$ bases with respect to $M$, see \eqref{ZEqnNum522860}. To achieve the partitioning, first the bases of $U_{B} $ are transformed by the $U_{CQT} $ is the quantum constant $Q$ transform \cite{qtr}. The $U_{CQT} $ operation is similar to the discrete QFT (quantum Fourier transform) transform \cite{ref22}, and defined in the following manner.

The $U_{CQT} $ transform is defined as
\begin{equation} \label{ZEqnNum312746} 
U_{CQT} \left({\left| k \right\rangle} ,m\right)={\textstyle\frac{1}{\sqrt{N} }} \sum _{j=0}^{N-1}f_{W} \left(j-m\right)e^{{2\pi ijQ\mathord{\left/ {\vphantom {2\pi ijQ m}} \right. \kern-\nulldelimiterspace} m} } {\left| j \right\rangle}  ={\left| \varphi _{k}  \right\rangle} ,  
\end{equation} 
where ${\left| k \right\rangle} $ is a quantum state of the computational basis $B$, and in the current setting
\begin{equation} \label{104)} 
N=K,  
\end{equation} 
and
\begin{equation} \label{105)} 
{\left| k \right\rangle} =E_{mk} ,  
\end{equation} 
thus $B$ is as
\begin{equation} \label{106)} 
B:\left\{{\left| 0 \right\rangle} ,\ldots ,{\left| K-1 \right\rangle} \right\},   
\end{equation} 
while $h$ is selected such that
\begin{equation} \label{107)} 
0\le \left(j-h\right)\le N-1=K-1 
\end{equation} 
holds, and $Q$ is defined via the following relation
\begin{equation} \label{108)} 
{\textstyle\frac{2\pi k}{K}} ={\textstyle\frac{2\pi Q}{h}} ,  
\end{equation} 
from which $Q$ is yielded at a given $h$, $k$ and $K$, as
\begin{equation} \label{109)} 
Q={\textstyle\frac{hk}{K}} ,  
\end{equation} 
while $f_{W} \left(\cdot \right)$ is a windowing function \cite{four} that localizes the wavefunctions of the quantum register, defined via parameter $h$ as
\begin{equation} \label{ZEqnNum133734} 
f_{W} \left(j-h\right)={\textstyle\frac{1}{2}} \left(1-\cos \left({\textstyle\frac{2\pi \left(h-m\right)}{K-1}} \right)\right).   
\end{equation} 
(Footnote: The function in \eqref{ZEqnNum133734} is the so-called Hanning window \cite{four}.) 

The ${\left| \varphi _{k}  \right\rangle} $ output states of $U_{CQT} $ therefore identify a set ${\rm {\mathcal S}}_{\varphi } $ of states, as
\begin{equation} \label{111)} 
{\rm {\mathcal S}}_{\varphi } :\left\{{\left| \varphi _{k}  \right\rangle} :k=0,\ldots ,K-1\right\} 
\end{equation} 
that formulates an orthonormal basis.

The $U_{CQT}^{\dag } $ inverse of $U_{CQT} $ will be processed as the $U_{P} $partitioning is completed, with the same $f_{W} \left(\cdot \right)$ windowing function, defined as
\begin{equation} \label{ZEqnNum324804} 
U_{CQT}^{\dag } \left({\left| k \right\rangle} ,h\right)={\textstyle\frac{1}{\sqrt{K} }} \sum _{j=0}^{K-1}f_{W} \left(j-h\right)e^{{-2\pi ijQ\mathord{\left/ {\vphantom {-2\pi ijQ h}} \right. \kern-\nulldelimiterspace} h} } {\left| j \right\rangle}  .  
\end{equation} 
Applying \eqref{ZEqnNum312746} on the $K$ estimated bases $\left\{E_{mk} \right\}$ yields the $C_{B} $ transformed bases, as 
\begin{equation} \label{ZEqnNum142233} 
\begin{split}
   {{C}_{B}}&={{U}_{CQT}}\left( {{U}_{B}} \right) \\ 
 & =\left\{ {{C}_{mk}} \right\}\in {{\mathbb{C}}^{M\times K}},  
\end{split}
\end{equation} 
where $C_{mk} $ is as,
\begin{equation} \label{114)} 
C_{mk} =U_{CQT} \left(E_{mk} \right).   
\end{equation} 
After the application of \eqref{ZEqnNum142233}, the resulting system is therefore as
\begin{equation} \label{ZEqnNum734223} 
C_{B} W=\left(U_{CQT} U_{B} \right)W, 
\end{equation} 
where $C_{B} W\in {\mathbb{C}}^{M\times T} $.

\subsubsection{Basis Partitioning Unitary}
\begin{theorem}
(Partitioning the bases of source systems.) The $Q$ transformed bases can be partitioned to $M$ partitions via the $U_{P} $ partitioning unitary operation.
\end{theorem}
\begin{proof}
As the $U_{CQT} $ transforms of the $\left\{E_{mk} \right\}$ basis estimations \eqref{ZEqnNum908888} are determined via $C_{B} $ \eqref{ZEqnNum142233}, the $Q$ transformed bases are partitioned to $M$ partitions via the $U_{P} $ unitary operation, as follows. 

Let the system state from \eqref{ZEqnNum734223} be denoted by
\begin{equation} \label{ZEqnNum989553} 
S=C_{B} W 
\end{equation} 
and let $\tilde{S}$ be the estimation of $S$ \cite{clus1}, defined as
\begin{equation} \label{ZEqnNum654394} 
\tilde{S}=\left\langle \left\langle {\rm {\mathcal R}{\mathcal E}}\right\rangle {\rm {\mathcal H}}\right\rangle ,   
\end{equation} 
where 
\begin{equation} \label{118)} 
{\rm {\mathcal T}}\in \left\{{\rm {\mathcal R}},{\rm {\mathcal E}},{\rm {\mathcal H}}\right\} 
\end{equation} 
is a tensor (multidimensional array) \cite{tens1,tens2} with dimension $\dim \left({\rm {\mathcal T}}\right)$, and size
\begin{equation} \label{119)} 
s\left({\rm {\mathcal T}}\right)=\prod _{i=1}^{\dim \left({\rm {\mathcal T}}\right)}\left|d_{i} \left({\rm {\mathcal T}}\right)\right| ,   
\end{equation} 
where $\left|d_{i} \left({\rm {\mathcal T}}\right)\right|$ is the size of the $i$-th dimension $d_{i} \left({\rm {\mathcal T}}\right)$. 

Let 
\begin{equation} \label{120)} 
{\rm {\mathcal R}}={\rm {\mathcal A}}\circ {\rm {\mathcal B}} 
\end{equation} 
be a translation tensor of size
\begin{equation} \label{121)}
\begin{split}
   s\left( \mathcal{R} \right)&=\prod\limits_{i=1}^{\dim\left( \mathcal{R} \right)}{\left| {{d}_{i}}\left( \mathcal{R} \right) \right|} \\ 
 & =\prod\limits_{i=1}^{\dim\left( \mathcal{A} \right)}{\left| {{d}_{i}}\left( \mathcal{A} \right) \right|}\times \prod\limits_{i=1}^{\dim\left( \mathcal{B} \right)}{\left| {{d}_{i}}\left( \mathcal{B} \right) \right|,}  
\end{split}
\end{equation}
with 
\begin{equation} \label{122)} 
\dim \left({\rm {\mathcal R}}\right)=3,   
\end{equation} 
as
\begin{equation} \label{123)} 
\left|d_{1} \left({\rm {\mathcal R}}\right)\right|=M,   
\end{equation} 
\begin{equation} \label{124)} 
\left|d_{2} \left({\rm {\mathcal R}}\right)\right|=1,   
\end{equation} 
and
\begin{equation} \label{125)} 
\left|d_{3} \left({\rm {\mathcal R}}\right)\right|=M 
\end{equation} 
and let
\begin{equation} \label{126)} 
{\rm {\mathcal E}}={\rm {\mathcal A}}\circ {\rm {\mathcal C}} 
\end{equation} 
be a tensor of size
\begin{equation} \label{127)}
\begin{split}
   s\left( \mathcal{E} \right)&=\prod\limits_{i=1}^{\dim\left( \mathcal{E} \right)}{\left| {{d}_{i}}\left( \mathcal{E} \right) \right|} \\ 
 & =\prod\limits_{i=1}^{\dim\left( \mathcal{A} \right)}{\left| {{d}_{i}}\left( \mathcal{A} \right) \right|}\times \prod\limits_{i=1}^{\dim\left( \mathcal{C} \right)}{\left| {{d}_{i}}\left( \mathcal{C} \right) \right|,}  
\end{split}
\end{equation} 
with
\begin{equation} \label{128)} 
\dim \left({\rm {\mathcal E}}\right)=2 
\end{equation} 
as
\begin{equation} \label{129)} 
\left|d_{1} \left({\rm {\mathcal E}}\right)\right|=M,    
\end{equation} 
\begin{equation} \label{130)} 
\left|d_{2} \left({\rm {\mathcal E}}\right)\right|=K, 
\end{equation} 
and with
\begin{equation} \label{131)} 
\dim \left({\rm {\mathcal H}}\right)=3,   
\end{equation} 
as
\begin{equation} \label{132)} 
\left|d_{1} \left({\rm {\mathcal R}}\right)\right|=1,   
\end{equation} 
\begin{equation} \label{133)} 
\left|d_{2} \left({\rm {\mathcal R}}\right)\right|=K 
\end{equation} 
and
\begin{equation} \label{134)} 
\left|d_{3} \left({\rm {\mathcal R}}\right)\right|=T,   
\end{equation} 
thus
\begin{equation} \label{135)} 
\dim \left({\rm {\mathcal A}}\right)=M 
\end{equation} 
and
\begin{equation} \label{136)} 
\dim \left({\rm {\mathcal B}}\right)=M 
\end{equation} 
while
\begin{equation} \label{137)} 
\dim \left({\rm {\mathcal C}}\right)=K.    
\end{equation} 
The term $\left\langle {\rm {\mathcal R}{\mathcal E}}\right\rangle $ is evaluated as
\begin{equation} \label{138)}
\begin{split}
  & {{\left\langle \mathcal{R}\mathcal{E} \right\rangle }_{\left\{ 1:\dim\left( \mathcal{A} \right),1:\dim\left( \mathcal{A} \right) \right\}}}\left( {{j}_{1}},\ldots ,{{j}_{\dim\left( \mathcal{B} \right)}},{{k}_{1}},\ldots ,{{k}_{\dim\left( \mathcal{C} \right)}} \right) \\ 
 & =\sum\limits_{{{i}_{1}}=1}^{{{d}_{1}}\left( \mathcal{A} \right)}{\cdots }\sum\limits_{{{i}_{\dim\left( \mathcal{A} \right)}}=1}^{{{d}_{\dim\left( \mathcal{A} \right)}}\left( \mathcal{A} \right)}{\mathcal{R}\left( {{i}_{1}},\ldots ,{{i}_{\dim\left( \mathcal{A} \right)}},{{j}_{1}},\ldots ,{{j}_{\dim\left( \mathcal{B} \right)}} \right)\mathcal{E}\left( {{i}_{1}},\ldots ,{{i}_{\dim\left( \mathcal{A} \right)}},{{k}_{1}},\ldots ,{{k}_{\dim\left( \mathcal{C} \right)}} \right),}  
\end{split}
\end{equation}
where ${\rm {\mathcal R}}\left(i,j\right)$ is the indexing for the elements of the tensor.

Let ${\rm {\mathcal E}}\left(\forall m,k\right)$ refer to the $j$-th column of ${\rm {\mathcal E}}$, and let ${\rm {\mathcal H}}\left(1,k,\forall t\right)$ refer to the $j$-th lateral slice of ${\rm {\mathcal H}}$. Then, let be a $U_{P} $ unitary operation that achieves the decomposition of \eqref{ZEqnNum654394} with respect to a given $k$, $k=1,\ldots ,K$, as   
\begin{equation} \label{ZEqnNum520467} 
\left[S\right]_{k} =\left\langle \left\langle {\rm {\mathcal R}{\mathcal E}}\left(\forall m,k\right)\right\rangle {\rm {\mathcal H}}\left(1,k,\forall t\right)\right\rangle  
\end{equation} 
with a particular cost function $f\left(U_{P} \right)$ of the $U_{P} $ unitary defined via the quantum relative entropy function, as
\begin{equation} \label{140)}
\begin{split}
   f\left( {{U}_{P}} \right)&=\underset{{\tilde{S}}}{\mathop{\min }}\,D\left( \left. {{\rho }_{S}} \right\|\tilde{S}{{{\tilde{S}}}^{\dagger }} \right) \\ 
 & =\underset{{\tilde{S}}}{\mathop{\min }}\,\text{Tr}\left( {{\rho }_{S}}\log \left( {{\rho }_{S}} \right) \right)-\text{Tr}\left( {{\rho }_{S}}\log \left( \tilde{S}{{{\tilde{S}}}^{\dagger }} \right) \right),  
\end{split}
\end{equation} 
where $\rho _{S} $ is the density matrix associated with $S$ is as in \eqref{ZEqnNum989553},
\begin{equation} \label{141)} 
\rho _{S} =U_{P} U_{CQT} U_{F} \left(\sum _{m=1}^{M}\sum _{t=1}^{T}\vec{X}^{\left(m,t\right)} \left(\vec{X}^{\left(m,t\right)} \right)^{\dag }   \right), 
\end{equation} 
while $\tilde{S}$ is given in \eqref{ZEqnNum654394}.

Using \eqref{ZEqnNum520467}, the $Q$-transformed bases are partitioned into $M$ classes, the partition $\Omega $ outputted by $U_{P} $ is evaluated as
\begin{equation} \label{ZEqnNum710735} 
\Omega =\arg \mathop{\max }\limits_{k} \left(\left[{\rm Q}\right]_{k} \right), 
\end{equation} 
where ${\rm Q}$ is a $1\times K$ size matrix, such that 
\begin{equation} \label{ZEqnNum128969} 
\left[{\rm Q}\right]_{k} =\sum _{m=1}^{M}\left\langle \left\langle {\rm {\mathcal R}}\left(\forall ,1,\forall \right){\rm {\mathcal E}}\left(\forall ,k\right)\right\rangle {\rm {\mathcal H}}\left(1,k,\forall \right)\right\rangle  .   
\end{equation} 
Since $M=2$ in our setting, the partition \eqref{ZEqnNum710735} can be rewritten as
\begin{equation} \label{144)} 
\Omega =\Omega _{Q}^{\left(1\right)} +\Omega _{Q}^{\left(2\right)} ,   
\end{equation} 
where $\Omega _{Q}^{\left(m\right)} $ identifies a cluster of $K_{m} $ $Q$-transformed bases for $m$-th system state, 
\begin{equation} \label{145)} 
\Omega _{Q}^{\left(m\right)} =\left\{\Omega _{Q}^{\left(m,k_{m} \right)} \right\}_{k_{m} =1}^{K_{m} } ,  
\end{equation} 
of 
\begin{equation} \label{146)} 
\left|\Omega _{Q}^{\left(m\right)} \right|=K_{m}  
\end{equation} 
bases formulated via the base estimations \eqref{ZEqnNum908888} for the $m$-th system state in \eqref{ZEqnNum522860}, such that 
\begin{equation} \label{147)} 
\sum _{m=1}^{M}K_{m}  =K.   
\end{equation} 
Since the partitioning is made over the $Q$ transformed bases, the output of $U_{P} $ is then transformed by the $U_{CQT}^{\dag } $ inverse transformation \eqref{ZEqnNum324804}.
\end{proof}

\subsubsection{Inverse Quantum Constant Q Transform}
Applying the $U_{CQT}^{\dag } $ inverse transformation \eqref{ZEqnNum324804} on the partitions \eqref{ZEqnNum128969} of the $Q$ transformed bases yields the decomposition of the bases of $U_{B} $ onto $M$ classes, as
\begin{equation} \label{148)} 
U_{CQT}^{\dag } \left(\Omega \right)=\theta =\sum _{m=1}^{M}\gamma ^{\left(m\right)} ,  
\end{equation} 
and since $M=2$ 
\begin{equation} \label{149)} 
\theta =\gamma ^{\left(1\right)} +\gamma ^{\left(2\right)} ,   
\end{equation} 
where $\gamma ^{\left(m\right)} $ identifies a cluster of $K_{m} $ bases for $m$-th system state.

Therefore, the resulting system state is as
\begin{equation} \label{ZEqnNum373293}
\begin{split}
  & U_{CQT}^{\dagger }\left( {{U}_{P}}\left( {{C}_{B}}W \right) \right) \\ 
 & =U_{CQT}^{\dagger }\left( {{U}_{P}}{{U}_{CQT}}{{U}_{B}} \right)W \\ 
 & =\chi W.  
\end{split}
\end{equation} 
The next problem is therefore the evaluation of the estimations of the $M=2$ source systems $\rho _{in} $ and $\zeta _{QR}^{\left(t\right)} $, as given in \eqref{ZEqnNum275626} from $\chi W$. Using the system state \eqref{ZEqnNum373293}, the system separation is produced by the $U_{{\rm DSTFT}}^{\dag } $ unitary that realizes the inverse quantum DSTFT (discrete short-time Fourier transform) \cite{four}.

\subsection{Inverse Quantum DSTFT and Quantum DFT}
The result of unitary $U_{ML} $ is evaluated further by the $U_{{\rm DSTFT}}^{\dag } $ unitary.
\begin{theorem}
 (Target source system recovery). Source system $m=1$ can be extracted by the $U_{{\rm DSTFT}}^{\dag } $ and $U_{DFT} $ discrete quantum Fourier transform on the output of an HRE quantum memory.
\end{theorem}
\begin{proof}
The $U_{{\rm DSTFT}}^{\dag } $ inverse quantum DSTFT transformation applied to a state ${\left| k \right\rangle} $ of the computational basis
\begin{equation} \label{151)} 
B:\left\{{\left| 0 \right\rangle} ,\ldots ,{\left| K-1 \right\rangle} \right\},   
\end{equation} 
is defined as
\begin{equation} \label{ZEqnNum876808} 
U_{{\rm DSTFT}}^{\dag } \left({\left| k \right\rangle} ,h\right)={\textstyle\frac{1}{\sqrt{K} }} \sum _{j=0}^{K-1}f_{W} \left(j-h\right)e^{{-2\pi ijk\mathord{\left/ {\vphantom {-2\pi ijk K}} \right. \kern-\nulldelimiterspace} K} } {\left| j \right\rangle}  ={\left| \psi _{k}  \right\rangle} ,  
\end{equation} 
where $h$ is selected such that
\begin{equation} \label{153)} 
0\le \left(j-h\right)\le K-1 
\end{equation} 
holds, set 
\begin{equation} \label{154)} 
{\rm {\mathcal S}}_{\psi } :\left\{{\left| \psi _{k}  \right\rangle} :k=0,\ldots ,K-1\right\} 
\end{equation} 
formulates an new orthonormal basis, while $f_{W} \left(\cdot \right)$ is a windowing function \cite{four} .

Using system state $\chi W$ in \eqref{ZEqnNum373293}, let $\gamma ^{\left(m,k\right)} $ be a $k$-th basis of cluster $\gamma ^{\left(m\right)} $, and let $\left(\chi W\right)^{\left(m,t\right)} $ be defined as
\begin{equation} \label{155)} 
\left(\chi W\right)^{\left(m,t\right)} =\left[\chi W\right]_{mk} =\sum _{k=1}^{K}\gamma ^{\left(m,k\right)} W_{k}^{\left(m,t\right)}   
\end{equation} 
and let system ${\left| \chi W \right\rangle} $ identify \eqref{ZEqnNum106419} as
\begin{equation} \label{ZEqnNum423661} 
{\left| \chi W \right\rangle} =\alpha \sum _{m=1}^{M}\sum _{k_{m} =1}^{K_{m} }{\left| k_{m}  \right\rangle}   ,   
\end{equation} 
where ${\left| k_{m}  \right\rangle} $ is the eigenvector of the Hamiltonian of $\gamma ^{\left(m,k_{m} \right)} $, $K_{m} $ is the cardinality of cluster $\gamma ^{\left(m\right)} $, while $\sum _{m=1}^{M}\sum _{k_{m} =1}^{K_{m} }\alpha   =1$. 

Since the ${\left| k_{1}  \right\rangle} $ values are some parameters of $U_{ML} $, we can redefine \eqref{ZEqnNum423661} as
\begin{equation} \label{ZEqnNum656752} 
{\left| \chi W \right\rangle} =\alpha \sum _{m=1}^{M}\sum _{k_{m} =1}^{K_{m} }{\left| k_{1} +x_{m,k_{m} }  \right\rangle}   ,  
\end{equation} 
where 
\begin{equation} \label{158)} 
x_{m,k_{m} } =\left\{\begin{array}{c} {0,{\rm \; if\; }m=1} \\ {!0,{\rm \; otherwise}} \end{array}\right. , 
\end{equation} 
and
\begin{equation} \label{159)} 
\alpha ={\textstyle\frac{1}{\sqrt{K} }} .   
\end{equation} 
In our setting, using $k_{m=1} $ as input parameter available from the $U_{ML} $ block, we redefine the formula of \eqref{ZEqnNum876808} via a unitary $\tilde{U}_{{\rm DSTFT}}^{\dag } $, as
\begin{equation} \label{ZEqnNum882164} 
\tilde{U}_{{\rm DSTFT}}^{\dag } \left({\left| k_{m}  \right\rangle} ,h\right)={\textstyle\frac{1}{\sqrt{K} }} \sum _{j=0}^{K-1}f_{W} \left(j-h\right)e^{{-2\pi ijk_{1} \mathord{\left/ {\vphantom {-2\pi ijk_{1}  K}} \right. \kern-\nulldelimiterspace} K} } {\left| j \right\rangle}  ={\left| \psi _{k_{m} }  \right\rangle} , 
\end{equation} 
where we set $f_{W} \left(j-h\right)$ to unity,
\begin{equation} \label{161)} 
f_{W} \left(j-h\right)=1.   
\end{equation} 
Thus, applying \eqref{ZEqnNum882164} on \eqref{ZEqnNum656752} yields
\begin{equation} \label{ZEqnNum613223}
\begin{split}
  & \tilde{U}_{\text{DSTFT}}^{\dagger }\left( \alpha \sum\limits_{m=1}^{M}{\sum\limits_{{{k}_{m}}=1}^{{{K}_{m}}}{\left| {{k}_{1}}+{{x}_{m,{{k}_{m}}}} \right\rangle }} \right) \\ 
 & =\tfrac{1}{\sqrt{K}}\sum\limits_{j=0}^{K-1}{\left( \left( \alpha \sum\limits_{m=0}^{M-1}{\sum\limits_{{{k}_{m}}=0}^{{{K}_{m}}-1}{{{e}^{{-2\pi ij{{k}_{1}}}/{K}\;}}}} \right){{e}^{{-2\pi ij{{x}_{m,{{k}_{m}}}}}/{K}\;}} \right)\left| j \right\rangle } \\ 
 & =\tfrac{1}{\sqrt{K}}\left( \sum\limits_{j=0}^{K-1}{{{e}^{{-2\pi ij{{x}_{m,{{k}_{m}}}}}/{K}\;}}} \right)\left( \alpha \sum\limits_{m=0}^{M-1}{\sum\limits_{{{k}_{m}}=0}^{{{K}_{m}}-1}{{{e}^{{-2\pi ij{{k}_{1}}}/{K}\;}}}} \right)\left| j \right\rangle ,  
\end{split}
\end{equation} 
where
\begin{equation} \label{163)} 
j={\textstyle\frac{K}{k_{m} }} ,   
\end{equation} 
and $\sum _{j=0}^{K-1}e^{{-2\pi ijx_{m,k_{m} } \mathord{\left/ {\vphantom {-2\pi ijx_{m,k_{m} }  K}} \right. \kern-\nulldelimiterspace} K} }  =1$, thus \eqref{ZEqnNum613223} can be rewritten as
\begin{equation} \label{ZEqnNum249034}
\begin{split}
  & \tilde{U}_{\text{DSTFT}}^{\dagger }\left( \alpha \sum\limits_{m=1}^{M}{\sum\limits_{{{k}_{m}}=1}^{{{K}_{m}}}{\left| {{k}_{1}}+{{x}_{m,{{k}_{m}}}} \right\rangle }} \right) \\ 
 & =\tfrac{1}{\sqrt{K}}\left( \alpha \sum\limits_{m=0}^{M-1}{\sum\limits_{{{k}_{m}}=0}^{{{K}_{m}}-1}{{{e}^{{-2\pi i\left( \tfrac{K}{{{k}_{m}}} \right){{k}_{1}}}/{K}\;}}}} \right)\left| \tfrac{K}{{{k}_{m}}} \right\rangle .  
\end{split}
\end{equation} 
As follows, if 
\begin{equation} \label{165)} 
j={\textstyle\frac{K}{k_{1} }} ,   
\end{equation} 
then, the resulting $\Pr \left(j\right)$ probability is
\begin{equation}\label{166)}
\begin{split}
   \Pr \left( j \right)&=\tfrac{1}{K}{{\left| \alpha \sum\limits_{k=0}^{{{K}_{1}}-1}{{{e}^{{-2\pi ij{{k}_{1}}}/{K}\;}}} \right|}^{2}} \\ 
 & =\tfrac{1}{K}{{\left| \alpha \sum\limits_{k=0}^{{{K}_{1}}-1}{{{e}^{{-2\pi i\tfrac{K}{{{k}_{1}}}{{k}_{1}}}/{K}\;}}} \right|}^{2}} \\ 
 & =\tfrac{1}{K}{{\left| \alpha  \right|}^{2}}K_{1}^{2} \\ 
 & =\tfrac{1}{{{K}^{2}}}K_{1}^{2},  
\end{split}
\end{equation} 
while for the remaining $j$-s, the probabilities are vanished out, thus
\begin{equation} \label{167)} 
\Pr \left(j\right)=0,  
\end{equation} 
if
\begin{equation} \label{168)} 
j\ne {\textstyle\frac{K}{k_{1} }} .  
\end{equation} 
Therefore, applying the $U_{DFT} $ discrete quantum Fourier transform on the resulting system state \eqref{ZEqnNum249034}, defined in our setting as
\begin{equation} \label{169)} 
U_{DFT} \left({\left| k \right\rangle} \right)={\textstyle\frac{1}{\sqrt{K_{1} } }} \sum _{j=0}^{K_{1} -1}e^{{2\pi ijk\mathord{\left/ {\vphantom {2\pi ijk K_{1} }} \right. \kern-\nulldelimiterspace} K_{1} } } {\left| j \right\rangle} ,  
\end{equation} 
yields the source system $m=1$ in terms of the $K_{1} $ bases, as
\begin{equation} \label{ZEqnNum896340}
\begin{split}
   {{U}_{DFT}}\tilde{U}_{\text{DSTFT}}^{\dagger }\left( \alpha \sum\limits_{m=1}^{M}{\sum\limits_{{{k}_{m}}=1}^{{{K}_{m}}}{\left| {{k}_{1}}+{{x}_{m,{{k}_{m}}}} \right\rangle }} \right)&=\tfrac{1}{\sqrt{{{K}_{1}}}}\sum\limits_{{{k}_{m}}=1}^{{{K}_{1}}}{\left| {{k}_{1}} \right\rangle } \\ 
 & =\left| {{\Phi }^{*}} \right\rangle ,  
\end{split}
\end{equation} 
that identifies the target system from \eqref{ZEqnNum423447}.

The proof is concluded here.
\end{proof}

The state of the $QR$ quantum register after the $\tilde{U}_{{\rm CQT}}^{\dag } $ operation and after the $\tilde{U}_{{\rm DSTFT}}^{\dag } $ operation is depicted in \fref{fig3}. 

 \begin{center}
\begin{figure*}[!htbp]
\begin{center}
\includegraphics[angle = 0,width=1\linewidth]{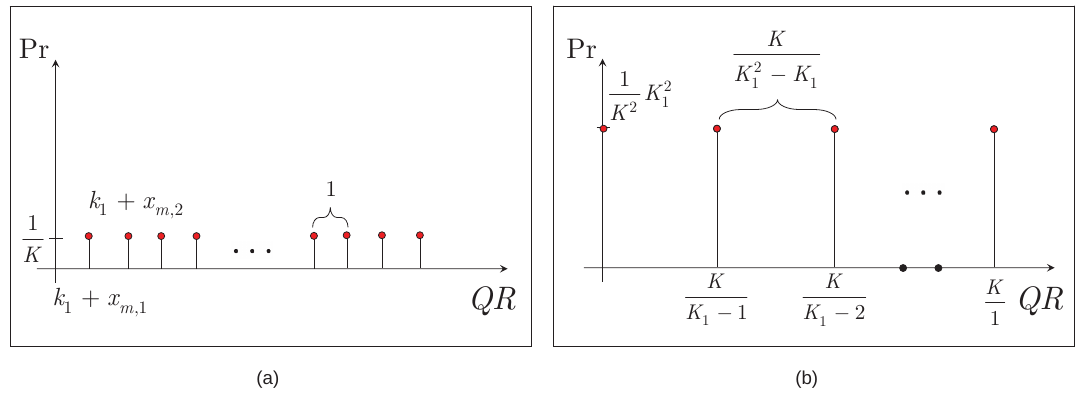}
\caption{(a) The state of the $QR$ quantum register after the $\tilde{U}_{{\rm CQT}}^{\dag } $ operation. The quantum register contains $K=\sum _{m}K_{m}  $ states, ${\left| k_{1} +x_{m,k_{m} }  \right\rangle} $, each with probability ${{\left| \alpha  \right|}^{2}}={1}/{K}\;$, with a unit distance between the states (depicted by the red dots). (b) The state of the $QR$ quantum register after the $\tilde{U}_{{\rm DSTFT}}^{\dag } $ operation. The quantum register contains $K_{1} $ quantum states, ${\left| {\textstyle\frac{K}{k_{1} }}  \right\rangle} $, $k_{1} =0,\ldots ,K_{1} -1$, each with probability ${\textstyle\frac{1}{K}} \left|\alpha \right|^{2} K_{1}^{2} ={\textstyle\frac{1}{K^{2} }} K_{1}^{2} $, with a distance ${\textstyle\frac{K}{K_{1}^{2} -K_{1} }} $ between the states (depicted by the red dots; the vanished-out states of the quantum register are depicted by the black dots).} 
 \label{fig3}
 \end{center}
\end{figure*}
\end{center}

\section{Retrieval Efficiency}
\label{sec4}
This section evaluates the retrieval efficiency of an HRE quantum memory in terms of the achievable output SNR values.
\begin{theorem}
(Retrieval efficiency of an HRE quantum memory). The SNR of the output quantum system of an HRE quantum memory is evolvable from the difference of the wave function energy ratios taken between the input system, the quantum register system, and the output quantum system.
\end{theorem}
\begin{proof}
Let ${\left| \psi _{in}  \right\rangle} $ be an arbitrary quantum system fed into the input of an HRE quantum memory unit,
\begin{equation} \label{171)} 
{\left| \psi _{in}  \right\rangle} =\sum _{i}a_{i} {\left| i \right\rangle}  ,   
\end{equation} 
and let ${\left| \varphi  \right\rangle} $ be the state outputted from the $QR$ quantum register, 
\begin{equation} \label{172)} 
{\left| \varphi  \right\rangle} =U_{QR} {\left| \psi _{in}  \right\rangle} ,    
\end{equation} 
where $U_{QG} $ is an unknown transformation.

Let ${\left| \Phi ^{*}  \right\rangle} $ be the output system of as given in \eqref{ZEqnNum896340}, that can be rewritten as
\begin{equation} \label{173)} 
{\left| \Phi ^{*}  \right\rangle} =U{\left| \varphi  \right\rangle} =U\left(U_{QR} {\left| \psi _{in}  \right\rangle} \right),  
\end{equation} 
where $U$ is the operator of the integrated unitary operations of the HRE quantum memory, defined as
\begin{equation} \label{174)} 
U=U_{ML} \tilde{U}_{{\rm DSTFT}}^{\dag } U_{DFT} =U_{F} U_{CQT} U_{P} U_{CQT}^{\dag } \tilde{U}_{{\rm DSTFT}}^{\dag } U_{DFT} .  
\end{equation} 
Then, let ${\rm {\mathcal O}}_{V} $ be a verification oracle that computes the energy $E$ of a wavefunction ${\left| \psi  \right\rangle} =\sum _{i}c_{i} {\left| \phi _{i}  \right\rangle} $ \cite{snr} as
\begin{equation} \label{175)} 
E\left( \psi  \right)=\frac{\int\limits{\left\langle  \psi  \right|\hat{H}\left| \psi  \right\rangle }}{\int\limits{\left\langle  \psi  | \psi  \right\rangle }}=\frac{\sum\nolimits_{ij}{c_{i}^{*}{{c}_{j}}\int\limits{\left\langle  {{\phi }_{i}} \right|\hat{H}\left| {{\phi }_{j}} \right\rangle }}}{\sum\nolimits_{ij}{c_{i}^{*}{{c}_{j}}\int\limits{\left\langle  {{\phi }_{i}} | {{\phi }_{j}} \right\rangle }}},
\end{equation} 
where $\hat{H}$ is a Hamiltonian. 

Then, let evaluate the corresponding energies of wavefunctions ${\left| \psi _{in}  \right\rangle} $, ${\left| \varphi  \right\rangle} $ and ${\left| \Phi ^{*}  \right\rangle} $ via ${\rm {\mathcal O}}_{V} $, as
\begin{equation} \label{ZEqnNum611323} 
S=E\left(\psi _{in} \right),  
\end{equation} 
\begin{equation} \label{177)} 
X=E\left(\varphi \right),  
\end{equation} 
and
\begin{equation} \label{ZEqnNum434181} 
T=E\left(\Phi ^{*} \right).  
\end{equation} 
Then, let $\Delta $ be the difference of the ratios of wavefunction energies, defined as
\begin{equation} \label{ZEqnNum172946} 
\Delta =R\left(S,T\right)-R\left(S,X\right) 
\end{equation} 
where
\begin{equation} \label{180)} 
R\left(S,T\right)={\textstyle\frac{S}{T}} ,  
\end{equation} 
and
\begin{equation} \label{181)} 
R\left(S,X\right)={\textstyle\frac{S}{X}} .  
\end{equation} 
From the quantities of \eqref{ZEqnNum611323}-\eqref{ZEqnNum434181}, let ${\rm SNR}\left({\left| \Phi ^{*}  \right\rangle} \right)$ be the SNR of the output system ${\left| \Phi ^{*}  \right\rangle} $, defined as
\begin{equation}\label{182)}
\begin{split}
   \text{SNR}\left( \left| {{\Phi }^{*}} \right\rangle  \right)&=10{{\log }_{10}}R\left( S,T \right) \\ 
 & ={{\log }_{10}}\Delta +\tfrac{1}{10}\text{SNR}\left( \left| X \right\rangle  \right),  
\end{split}
\end{equation} 
where 
\begin{equation} \label{183)} 
{\rm SNR}\left({\left| X \right\rangle} \right)=10\log _{10} R\left(S,X\right),  
\end{equation} 
while $\Delta $ is as given in \eqref{ZEqnNum172946}.

Therefore, the SNR of the output system can be evolved from the difference of the ratios of the wavefunction energies as
\begin{equation} \label{184)}
\begin{split}
   \text{SNR}\left( \left| {{\Phi }^{*}} \right\rangle  \right)&=10{{\log }_{10}}R\left( S,T \right) \\ 
 & =10\left( {{\log }_{10}}\Delta +{{\log }_{10}}R\left( S,X \right) \right) \\ 
 & =10\left( {{\log }_{10}}\left( R\left( S,T \right)-R\left( S,X \right) \right)+{{\log }_{10}}R\left( S,X \right) \right) \\ 
 & =10\left( {{\log }_{10}}\tfrac{R\left( S,T \right)}{R\left( S,X \right)}+{{\log }_{10}}R\left( S,X \right) \right).  
\end{split}
\end{equation}
It also can be verified that $\Delta $ from \eqref{ZEqnNum172946} can be rewritten as
\begin{equation} \label{185)} 
\Delta =10^{{\Delta _{{\rm SNR}} \mathord{\left/ {\vphantom {\Delta _{{\rm SNR}}  10}} \right. \kern-\nulldelimiterspace} 10} } ,   
\end{equation} 
where $\Delta _{{\rm SNR}} $ is an SNR difference, defined as
\begin{equation} \label{186)} 
\Delta _{{\rm SNR}} ={\rm SNR}\left({\left| \Phi ^{*}  \right\rangle} \right)-{\rm SNR}\left({\left| X \right\rangle} \right).   
\end{equation} 
The high SNR values are reachable at moderate values of wavefunction energy ratio differences \eqref{ZEqnNum172946}, therefore a high retrieval efficiency (high SNR values) can be produced by the local unitaries of the memory unit (see also \fref{fig6}). 

The proof is concluded here.
\end{proof}

The verification of the retrieval efficiency of the output of an HRE quantum memory unit is depicted in \fref{fig5}.

 \begin{center}
\begin{figure*}[!htbp]
\begin{center}
\includegraphics[angle = 0,width=1\linewidth]{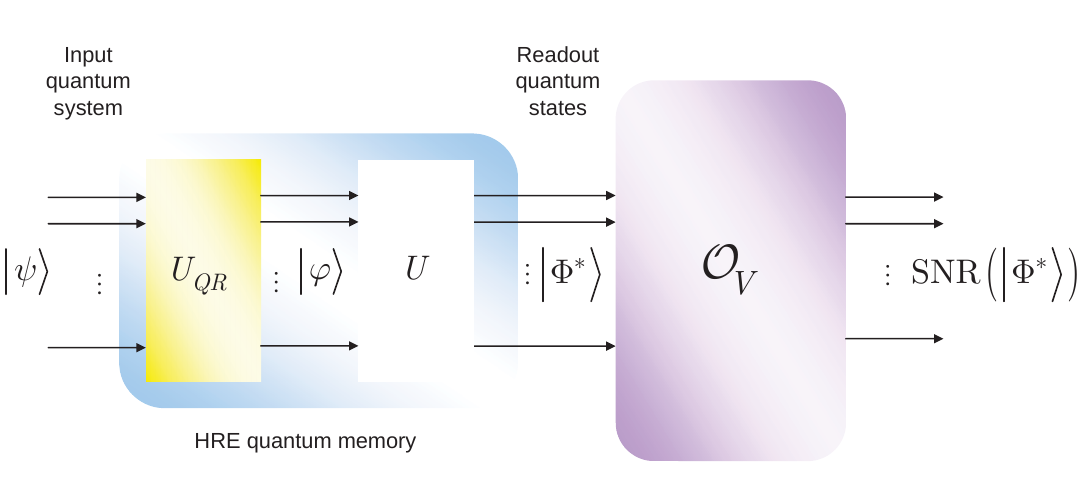}
\caption{Verification of the retrieval efficiency of an HRE quantum memory unit via an ${\rm {\mathcal O}}_{V} $ verification oracle. In the verification procedure, an unknown quantum system ${\left| \psi  \right\rangle} $ is stored in the $QR$ quantum register that is evolved by an unknown operation $U_{QR} $ of the $QR$ quantum register. The output of $QR$ is an unknown quantum system ${\left| \varphi  \right\rangle} $ that is processed further by the $U$ integrated unitary operations of the HRE quantum memory. The output system of the HRE quantum memory is ${\left| \Phi ^{*}  \right\rangle} $ \eqref{ZEqnNum896340}. The ${\rm {\mathcal O}}_{V} $ oracle evaluates the ${\rm SNR}$ of the readout quantum system ${\left| \Phi ^{*}  \right\rangle} $.} 
 \label{fig5}
 \end{center}
\end{figure*}
\end{center}

The output SNR values in the function of the $\Delta $ wave function energy ratio difference are depicted in \fref{fig6}. 

 \begin{center}
\begin{figure*}[!htbp]
\begin{center}
\includegraphics[angle = 0,width=1\linewidth]{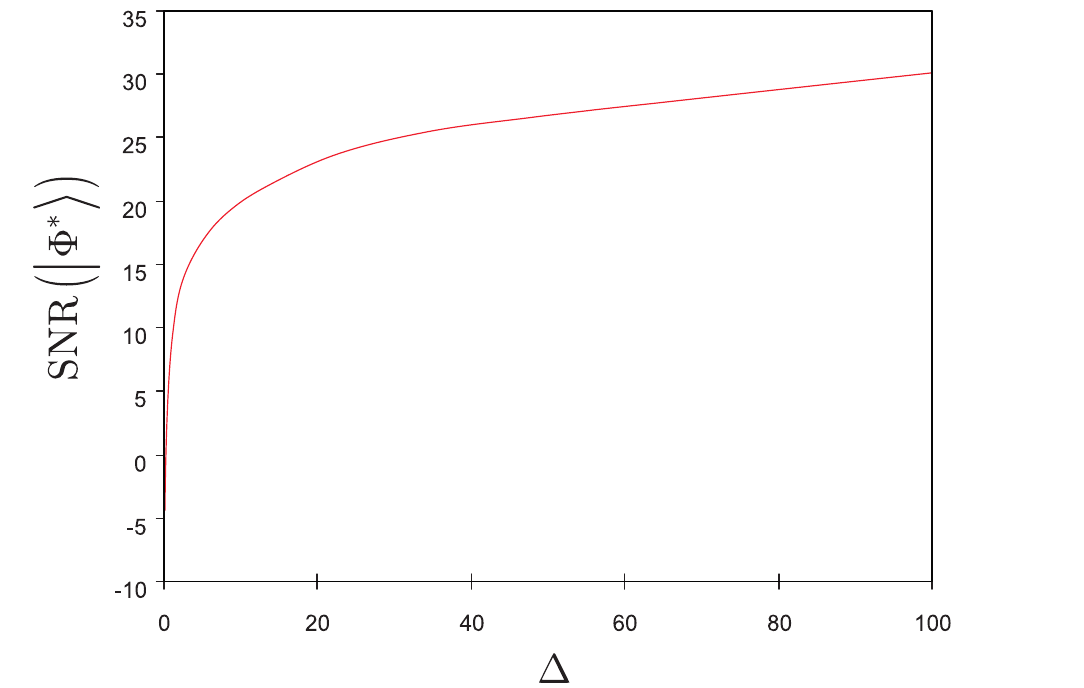}
\caption{The output SNR values, ${\rm SNR}\left({\left| \Phi ^{*}  \right\rangle} \right)=10\log _{10} R\left(S,T\right)$, of an HRE quantum memory in the function of $\Delta =R\left(S,T\right)-R\left(S,X\right)$, where $R\left(S,T\right)={\textstyle\frac{S}{T}} $, $R\left(S,X\right)={\textstyle\frac{S}{X}} $, $S=E\left(\psi _{in} \right)$, $X=E\left(\varphi \right)$, and $T=E\left(\Phi ^{*} \right)$.} 
 \label{fig6}
 \end{center}
\end{figure*}
\end{center}

\section{Conclusions}
\label{sec5}
Quantum memories are a cornerstone of the construction of quantum computers and a high-performance global-scale quantum Internet. Here, we defined the HRE quantum memory for near-term quantum devices. We defined the unitary operations of an HRE quantum memory and proved the learning procedure. We showed that the local unitaries of an HRE quantum memory integrates a group of quantum machine learning operations for the evaluation of the unknown quantum system, and a group of unitaries for the target system recovery. We determined the achievable output SNR values. The HRE quantum memory is a particularly convenient unit for gate-model quantum computers and the quantum Internet.

\section*{Acknowledgements}
The research reported in this paper has been supported by the Hungarian Academy of Sciences (MTA Premium Postdoctoral Research Program 2019), by the National Research, Development and Innovation Fund (TUDFO/51757/2019-ITM, Thematic Excellence Program), by the National Research Development and Innovation Office of Hungary (Project No. 2017-1.2.1-NKP-2017-00001), by the Hungarian Scientific Research Fund - OTKA K-112125 and in part by the BME Artificial Intelligence FIKP grant of EMMI (Budapest University of Technology, BME FIKP-MI/SC).


\newpage
\appendix
\setcounter{table}{0}
\setcounter{figure}{0}
\setcounter{equation}{0}
\setcounter{algocf}{0}
\renewcommand{\thetable}{\Alph{section}.\arabic{table}}
\renewcommand{\thefigure}{\Alph{section}.\arabic{figure}}
\renewcommand{\theequation}{\Alph{section}.\arabic{equation}}
\renewcommand{\thealgocf}{\Alph{section}.\arabic{algocf}}

\setlength{\arrayrulewidth}{0.1mm}
\setlength{\tabcolsep}{5pt}
\renewcommand{\arraystretch}{1.5}

\section{Appendix}
\subsection{Abbreviations}
\begin{description}
\item[DFT] Discrete Fourier Transform
\item[DSTFT] Discrete Short-Time Fourier Transform
\item[HRE] High-Retrieval Efficiency
\item[SNR] Signal-to-Noise Ratio
\end{description}

\subsection{Notations}
The notations of the manuscript are summarized in \tref{tab2}.
\begin{center}
\begin{longtable}{||l|p{4.5in}||}
\caption{Summary of notations.}
\label{tab2}
\endfirsthead
\endhead
\hline
\textit{Notation} & \textit{Description} \\ \hline
$\rho _{in} $ & An unknown input quantum system formulated by $n$ unknown density matrices. \\ \hline 
$\rho _{i} $ & An $i$-th density matrix, $i=1,\ldots ,n$.  \\ \hline 
$QR$ & Quantum register of an HRE quantum memory. \\ \hline 
$\sigma _{QR} $ & Mixed state of the $QR$ quantum register. \\ \hline 
$T$ & Total evolution time, $t=1,\ldots ,T$.  \\ \hline 
$\sigma _{QR}^{\left(t\right)} $ & Mixed state of the $QR$ quantum register at a given $t$, $\sigma _{QR}^{\left(t\right)} \in {\mathbb{C}}$, $\sigma _{QR}^{\left(t\right)} =\sum _{i=1}^{n}\lambda _{i}^{\left(t\right)} {\left| \varphi _{i}^{\left(t\right)}  \right\rangle} {\left\langle \varphi _{i}^{\left(t\right)}  \right|}  $, $t=1,\ldots ,T$. \\ \hline 
$U_{QR}^{\left(t\right)} $ & Unknown evolution matrix of the $QR$ quantum register at a given $t$. \\ \hline 
$\dim \left(U_{QR}^{\left(t\right)} \right)$ & Dimension of $U_{QR}^{\left(t\right)} $, $\dim \left(U_{QR}^{\left(t\right)} \right)=d^{n} \times d^{n} $, where $d$ is the dimension of the quantum system. \\ \hline 
$X_{i}^{\left(t\right)} $ & A complex quantity, defined as $X_{i}^{\left(t\right)} =\sqrt{\lambda _{i}^{\left(t\right)} } {\left| \varphi _{i}^{\left(t\right)}  \right\rangle} $, $i=1,\ldots ,n$, $t=1,\ldots ,T$. \\ \hline 
$X^{\left(t\right)} $ & Sum of $n$ complex quantities, $X^{\left(t\right)} =\sum _{i=1}^{n}X_{i}^{\left(t\right)}  $. \\ \hline 
$\zeta _{QR}^{\left(t\right)} $ & An unknown residual density matrix at a given $t$, it formulates the mixed system of the quantum register as $\sigma _{QR}^{\left(t\right)} =\rho _{in} +\zeta _{QR}^{\left(t\right)} .$ \\ \hline 
$M$  & Number of source systems of the mixed quantum register, $\sigma _{QR}^{\left(t\right)} =\sum _{m=1}^{M}\rho _{m}  $, where $\rho _{m} $ is an $m$-th source system, $m=1,\ldots ,M$. \\ \hline 
$\rho _{m} $ & An $m$-th source density matrix of the mixed quantum register state, $m=1,\ldots ,M$. \\ \hline 
$X_{i}^{\left(m,t\right)} $ & A complex quantity associated with an $m$-th source system, $X_{i}^{\left(m,t\right)} =\sqrt{\lambda _{i}^{\left(m,t\right)} } {\left| \varphi _{i}^{\left(m,t\right)}  \right\rangle} ,$ $m=1,\ldots ,M$, $i=1,\ldots ,n$, $t=1,\ldots ,T$. \\ \hline 
$X^{\left(m,t\right)} $ & Sum of $n$ complex quantities, $X^{\left(m,t\right)} =\sum _{i=1}^{n}X_{i}^{\left(m,t\right)}  $, $m=1,\ldots ,M$, $t=1,\ldots ,T$. \\ \hline 
$\tilde{X}^{\left(m,t\right)} $ & An approximation of $X^{\left(m,t\right)} $. \\ \hline 
$U_{QR} $ & Unknown transformation matrix of the $QR$ quantum register over the total evolution time $T$. \\ \hline 
$V_{QG} $ & Inverse matrix of the unknown $U_{QR} $. \\ \hline 
$\sigma _{out} $ & Output quantum system. \\ \hline 
$U_{ML} $ & Unitary of a quantum machine learning procedure. \\ \hline 
$U_{F} $ & Factorization unitary, evaluates $K$ bases for the source system decomposition, and defines a $W$ auxiliary quantum system. \\ \hline 
$U_{CQT} $ & Unitary of the quantum constant $Q$ transform. The $U_{CQT} $ transform is a preliminary operation for the partitioning of the $K$ bases onto $M$ clusters via unitary $U_{P} $. \\ \hline 
$f_{W} $ & A windowing function in $U_{CQT} $. \\ \hline 
$U_{P} $ & A basis partitioning unitary that clusters the bases with respect to the $M$ source systems. \\ \hline 
$U_{CQT}^{\dag } $ & A unitary, inverse of $U_{CQT} $. \\ \hline 
$\tilde{U}_{{\rm DSTFT}}^{\dag } $ & Unitary of the inverse quantum DSTFT (discrete short-time Fourier transform) operation. \\ \hline 
$U_{DFT} $ & Quantum discrete Fourier transform. \\ \hline 
$u_{mk} $ & Parameter in the basis estimation procedure of $U_{ML} $, $m=1,\ldots ,M$, $k=1,\ldots ,K$. \\ \hline 
$H_{mk} $ & A Hamiltonian, $H_{mk} =G_{mk} {\left| k_{m}  \right\rangle} {\left\langle k_{m}  \right|} $, where $G_{mk} $ is the eigenvalue of basis ${\left| k_{m}  \right\rangle} $, $H_{mk} {\left| k_{m}  \right\rangle} =G_{mk} {\left| k_{m}  \right\rangle} $. \\ \hline 
$\tau $ & Application time, sub-parameter of $u_{mk} $. \\ \hline 
$w_{kt} $ & A system state evolved via $U_{ML} $. \\ \hline 
$K$  & Number of bases evolved via $U_{ML} $. \\ \hline 
$T$ & Total evolution time of the quantum system in the quantum register. \\ \hline 
$M$ & Number of source systems of the mixed quantum system $\sigma _{QR} $ of the quantum register. \\ \hline 
$U_{B} $ & A complex basis matrix, $U_{B} =\left\{u_{mk} \right\}\in {\mathbb{C}}^{M\times K} $, $m=1,\ldots ,M$, $k=1,\ldots ,K$. \\ \hline 
$\vec{\rho }_{W} $ & A complex matrix, $\vec{\rho }_{W} \in {\mathbb{C}}^{K\times T} $, $\vec{\rho }_{W} =\left\{\rho _{W}^{\left(t\right)} \right\}_{t=1}^{T} $. \\ \hline 
$W_{k}^{\left(t\right)} $ & A complex quantity, $W_{k}^{\left(t\right)} \in {\mathbb{C}}$, $W_{k}^{\left(t\right)} =\sqrt{v_{k}^{\left(t\right)} } {\left| \phi _{k}  \right\rangle} $, $m=1,\ldots ,M$, $k=1,\ldots ,K$. \\ \hline 
$w_{kt} $ & A complex quantity, $w_{kt} =W_{k}^{\left(t\right)} $, $m=1,\ldots ,M$, $k=1,\ldots ,K$. \\ \hline 
$W$ & A complex matrix, $W=\left\{W_{k}^{\left(t\right)} =w_{kt} \right\}\in {\mathbb{C}}^{K\times T} $. \\ \hline 
$\vec{X}$ & A complex matrix, $\vec{X}\in {\mathbb{C}}^{M\times T} $. \\ \hline 
$\tilde{X}$ & A complex matrix, $\tilde{X}\in {\mathbb{C}}^{M\times T} $, an approximation of $\vec{X}$, as $\tilde{X}=U_{B} W$. \\ \hline 
$K_{m} $ & The number of bases associated with the $m$-th source system, $m=1,\ldots ,M$, $\sum _{m}K_{m}  =K$. \\ \hline 
${\left| \Phi ^{*}  \right\rangle} $ & Target output system state, ${\left| \Phi ^{*}  \right\rangle} ={\textstyle\frac{1}{\sqrt{K_{1} } }} \sum _{k_{1} =1}^{K_{1} }{\left| k_{1}  \right\rangle}  $, where $K_{1} $ is the number of bases for the source system $m=1$, $k_{1} =1,\ldots ,K_{1} $. \\ \hline 
$\rho _{\vec{X}} $ & A density matrix associated with $\vec{X}$, $\rho _{\vec{X}} =\sum _{m=1}^{M}\sum _{t=1}^{T}\vec{X}^{\left(m,t\right)} \left(\vec{X}^{\left(m,t\right)} \right)^{\dag }   $. \\ \hline 
$\rho _{\tilde{X}} $ & A density matrix associated with $\tilde{X}$, $\rho _{\tilde{X}} =\sum _{m=1}^{M}\sum _{t=1}^{T}\tilde{X}^{\left(m,t\right)} \left(\tilde{X}^{\left(m,t\right)} \right)^{\dag }   $. \\ \hline 
$D\left(\left. \cdot \right\| \cdot \right)$ & Quantum relative entropy function. \\ \hline 
$f\left(U_{F} \right)$ & Objective function of unitary $U_{F} $. \\ \hline 
${\rm {\mathcal L}}\left(\cdot \right)$ & A likelihood function. \\ \hline 
$\alpha _{mk} $ & A control parameter defined for $u_{mk} $, such that $u_{mk} \simeq \alpha _{mk} e^{-\alpha _{mk} u_{mk} } $. \\ \hline 
$\beta _{kt} $ & A control parameter defined for $w_{kt} $, such that $w_{kt} \simeq \beta _{kt} e^{-\beta _{kt} w_{kt} } $. \\ \hline 
$\zeta $ & A set of model parameters, $\zeta =\left\{U_{B} ,W\right\}$. \\ \hline 
$\tau _{mk}^{\left(t\right)} $ & A set of control parameters, $\tau _{mk}^{\left(t\right)} =\left\{\alpha _{mk} ,\beta _{kt} \right\}$. \\ \hline 
$\tilde{\zeta }$ & A maximum likelihood estimation of $\zeta $. \\ \hline 
${\rm {\mathcal D}}\left(\cdot \right)$ & A probability distribution. \\ \hline 
$\kappa _{mkt} $ & An estimation coefficient, $\kappa _{mkt} =u_{mk} w_{kt} $. \\ \hline 
$\vec{\kappa }$ & A complex matrix, $\vec{\kappa }\in {\mathbb{C}}^{M\times T} $, $\vec{\kappa }=\left\{\kappa ^{\left(1,t\right)} ,\ldots ,\kappa ^{\left(M,t\right)} \right\}_{t=1}^{T} $, where $\kappa ^{\left(m,t\right)} =\left(\kappa _{k=1}^{\left(m,t\right)} ,\ldots ,\kappa _{k=K}^{\left(m,t\right)} \right)^{T} ,$ with $\kappa _{k}^{\left(m,t\right)} =\kappa _{mkt} $. \\ \hline 
${\rm {\mathcal D}}_{v} \left(\cdot \right)$ & A variational distribution. \\ \hline 
$H\left({\rm {\mathcal D}}_{v} \left(\cdot \right)\right)$ & Entropy of a variational distribution ${\rm {\mathcal D}}_{v} \left(\cdot \right)$. \\ \hline 
${\rm {\mathcal L}}_{{\rm {\mathcal D}}_{v} } $ & A likelihood function. \\ \hline 
${\mathbb{E}}_{{\rm {\mathcal D}}_{v} \left(i\ne \Phi \right)} \left(\cdot \right)$ & Expectation function of the ${\rm {\mathcal D}}_{v} \left(i\right)$ variational distribution of $i$, such that $i\ne \Phi $, $\Phi \in \left\{\vec{\kappa },U_{B} ,W\right\}$,\newline ${\mathbb{E}}_{a} \left(f\left(a\right)+g\left(a\right)\right)={\mathbb{E}}_{a} \left(f\left(a\right)\right)+{\mathbb{E}}_{a} \left(g\left(a\right)\right)$,\newline for some functions $f\left(a\right)$ and $g\left(a\right)$, and\newline ${\mathbb{E}}_{a} \left(bf\left(a\right)\right)=b{\mathbb{E}}_{a} \left(f\left(a\right)\right)$\newline for some constant $b$. \\ \hline 
$f_{\delta } \left(\cdot \right)$ & Dirac delta function. \\ \hline 
$f_{\Gamma } \left(\cdot \right)$ & Gamma function,\newline $f_{\Gamma } \left(x\right)=\int\limits_{0}^{\infty }t^{x-1} e^{-t} dt $.  \\ \hline 
${\rm {\mathcal M}}$ & A multinomial distribution. \\ \hline 
$\eta _{mkt} $ & A multinomial parameter. \\ \hline 
$\eta _{k}^{\left(m,t\right)} $ & A multinomial parameter vector, $\eta _{k}^{\left(m,t\right)} =\left(\eta _{k=1}^{\left(m,t\right)} ,\ldots ,\eta _{k=K}^{\left(m,t\right)} \right)^{T} $, \newline $\sum _{k=1}^{K}\eta _{k}^{\left(m,t\right)}  =1$. \\ \hline 
${\rm {\mathcal G}}\left(\cdot \right)$ & A Gamma distribution, \newline ${\rm {\mathcal G}}\left(x;a,b\right)=e^{\left(a-1\right)\log x-\frac{x}{b} -\log f_{\Gamma } \left(a\right)-a\log b} ,$\newline where $a$ is a shape parameter, while $b$ is a scale parameter. \\ \hline 
$f_{\Gamma } \left(\cdot \right)$ & Gamma function. \\ \hline 
$H\left({\rm {\mathcal G}}\left(\cdot \right)\right)$ & Entropy of Gamma distribution ${\rm {\mathcal G}}\left(\cdot \right)$. \\ \hline 
$\partial _{{\rm {\mathcal G}}_{\log } } \left(\cdot \right)$ & Derivative of the log gamma function (digamma function),  $\partial _{{\rm {\mathcal G}}_{\log } } \left(x\right)={\textstyle\frac{d\log f_{\Gamma } \left(x\right)}{dx}} $. \\ \hline 
${\mathbb{E}}\left(\kappa _{mkt} \right)$ & Expected value of $\kappa _{mkt} $, ${\mathbb{E}}\left(\kappa _{mkt} \right)=X^{\left(m,t\right)} \eta _{mkt} $. \\ \hline 
$\tilde{\alpha }_{mk} \left(A\right)$ & Control parameter for $U_{B} $, $\tilde{\alpha }_{mk} \left(A\right)=1+\sum _{t=1}^{T}{\mathbb{E}}\left(\kappa _{mkt} \right). $ \\ \hline 
$\tilde{\alpha }_{mk} \left(B\right)$ & Control parameter for $U_{B} $, $\tilde{\alpha }_{mk} \left(B\right)={\textstyle\frac{1}{\sum _{t=1}^{T}{\mathbb{E}}\left(w_{kt} \right)+\alpha _{mk}  }} .$ \\ \hline 
$\tilde{\beta }_{kt} \left(A\right)$ & Control parameter for $W$, $\tilde{\beta }_{kt} \left(A\right)=1+\sum _{m=1}^{M}{\mathbb{E}}\left(\kappa _{mkt} \right). $ \\ \hline 
$\tilde{\beta }_{kt} \left(B\right)$ & Control parameter for $W$, $\tilde{\beta }_{kt} \left(B\right)={\textstyle\frac{1}{\sum _{m=1}^{M}{\mathbb{E}}\left(u_{mk} \right)+\beta _{kt}  }} .$ \\ \hline 
${\mathbb{E}}\left(w_{kt} \right)$ & Expected value of ${\mathbb{E}}\left(w_{kt} \right)$, ${\mathbb{E}}\left(w_{kt} \right)=\tilde{\beta }_{kt} \left(A\right)\tilde{\beta }_{kt} \left(B\right)$. \\ \hline 
${\mathbb{E}}\left(\log w_{kt} \right)$ & Expected value of ${\mathbb{E}}\left(\log w_{kt} \right)$, ${\mathbb{E}}\left(\log w_{kt} \right)=\partial _{{\rm {\mathcal G}}_{\log } } \left(\tilde{\beta }_{kt} \left(A\right)\right)+\log \tilde{\beta }_{kt} \left(B\right)$. \\ \hline 
${\mathbb{E}}\left(u_{mk} \right)$ & Expected value of ${\mathbb{E}}\left(u_{mk} \right)$, ${\mathbb{E}}\left(u_{mk} \right)=\tilde{\alpha }_{mk} \left(A\right)\tilde{\alpha }_{mk} \left(B\right)$. \\ \hline 
${\mathbb{E}}\left(\log u_{mk} \right)$ & Expected value of ${\mathbb{E}}\left(\log u_{mk} \right)$, ${\mathbb{E}}\left(\log u_{mk} \right)=\partial _{{\rm {\mathcal G}}_{\log } } \left(\tilde{\alpha }_{mk} \left(A\right)\right)+\log \tilde{\alpha }_{mk} \left(B\right)$. \\ \hline 
$E_{mk} $ & A basis estimation, $E_{mk} \approx \alpha _{mk} $. \\ \hline 
$F_{kt} $ & A system state estimation, $F_{kt} \approx \beta _{kt} $. \\ \hline 
$\tilde{\tau }_{k}^{\left(t\right)} $ & Estimation of the control parameters $\alpha _{mk} ,\beta _{kt} $, as $\tilde{\tau }_{mk}^{\left(t\right)} =\left\{E_{mk} ,F_{kt} \right\}$. \\ \hline 
$K^{*} $ & An optimal number of bases, $K^{*} =\arg \mathop{\max }\limits_{K} {\rm {\mathcal L}}_{{\rm {\mathcal D}}_{v} } \left(K\right)$, where ${\rm {\mathcal L}}_{{\rm {\mathcal D}}_{v} } \left(K\right)$ is the likelihood function ${\rm {\mathcal L}}_{{\rm {\mathcal D}}_{v} } $ at a particular base number $K$. \\ \hline 
$Q$ & Parameter of $U_{CQT} $. \\ \hline 
$f_{W} \left(\cdot \right)$ & Windowing function for $U_{CQT} $. \\ \hline 
$h$  & Parameter of $f_{W} \left(\cdot \right)$. \\ \hline 
$C_{B} $ & A complex matrix of transformed bases, $C_{B} =U_{CQT} \left(U_{B} \right)=\left\{C_{mk} \right\}\in {\mathbb{C}}^{M\times K} .$ \\ \hline 
$C_{mk} $ & A $Q$-transformed basis estimation parameter, $C_{mk} =U_{CQT} \left(E_{mk} \right)$. \\ \hline 
$S$ & A complex matrix, $S=C_{B} W$, $S\in {\mathbb{C}}^{M\times T} $, $C_{B} W=\left(U_{CQT} U_{B} \right)W.$ \\ \hline 
${\rm {\mathcal T}}$ & A tensor (multidimensional array). \\ \hline 
$\dim \left({\rm {\mathcal T}}\right)$ & Dimension of tensor ${\rm {\mathcal T}}$. \\ \hline 
${\rm {\mathcal R}}$ & A translation tensor. \\ \hline 
$f\left(U_{P} \right)$ & A cost function of $U_{P} $. \\ \hline 
$\Omega $ & A sum of partitioned bases. \\ \hline 
$\Omega _{Q}^{\left(m\right)} $ & A cluster of $\left|\Omega _{Q}^{\left(m\right)} \right|$ $Q$-transformed bases for the $m$-th system state, $\Omega _{Q}^{\left(m\right)} =\left\{\Omega _{Q}^{\left(m,k_{m} \right)} \right\}_{k_{m} =1}^{K_{m} } $,  $m=1,\ldots ,M$. \\ \hline 
$\left|\Omega _{Q}^{\left(m\right)} \right|$ & Cardinality of cluster $\Omega _{Q}^{\left(m\right)} $, $\left|\Omega _{Q}^{\left(m\right)} \right|=K_{m} $. \\ \hline 
$\theta $ & Partitioned bases transformed by $U_{CQT}^{\dag } $. \\ \hline 
$\gamma ^{\left(m\right)} $ & A cluster of $K_{m} $ bases for $m$-th system state. \\ \hline 
$\chi W$ & A system state, defined as $\chi W=U_{CQT}^{\dag } \left(U_{P} \left(C_{B} W\right)\right)=U_{CQT}^{\dag } \left(U_{P} U_{CQT} U_{B} \right)$. \\ \hline 
${\left| k_{m}  \right\rangle} $ & A basis state associated with the $m$-th source.  \\ \hline 
$x_{m,k_{m} } $ & Model parameter, defined as \newline $x_{m,k_{m} } =\left\{\begin{array}{c} {0,{\rm \; if\; }m=1} \\ {!0,{\rm \; otherwise}} \end{array}\right. .$ \\ \hline 
${\left| \psi _{in}  \right\rangle} $ & An arbitrary input system.  \\ \hline 
${\left| \Phi ^{*}  \right\rangle} $ & An output system, ${\left| \Phi ^{*}  \right\rangle} =U{\left| \varphi  \right\rangle} =U\left(U_{QR} {\left| \psi _{in}  \right\rangle} \right)$, where $U$ is the operator of the integrated unitary operations of the HRE quantum memory, defined as\newline $U=U_{ML} \tilde{U}_{{\rm DSTFT}}^{\dag } U_{DFT} =U_{F} U_{CQT} U_{P} U_{CQT}^{\dag } \tilde{U}_{{\rm DSTFT}}^{\dag } U_{DFT} $. \\ \hline 
${\rm {\mathcal O}}_{V} $ & A verification oracle that computes the energy $E$ of a wavefunction ${\left| \psi  \right\rangle} =\sum _{i}c_{i} {\left| \phi _{i}  \right\rangle}  $. \\ \hline 
$E\left(\psi \right)$ & Energy $E$ of a wavefunction ${\left| \psi  \right\rangle} =\sum _{i}c_{i} {\left| \phi _{i}  \right\rangle}  $. \\ \hline 
$\Delta $ & Wavefunction energy ratio difference, $\Delta =R\left(S,T\right)-R\left(S,X\right)$, where $R\left(S,T\right)={\textstyle\frac{S}{T}} $, $R\left(S,X\right)={\textstyle\frac{S}{X}} $, and $S=E\left(\psi _{in} \right)$, $X=E\left(\varphi \right)$, and $T=E\left(\Phi ^{*} \right)$. \\ \hline 
$\Delta _{{\rm SNR}} $ & An SNR difference. \\ \hline
\end{longtable}
\end{center}
\end{document}